\documentclass[12pt,onecolumn,draftclsnofoot]{IEEEtran}
\usepackage{float}
\usepackage{latexsym}
\usepackage[T1]{fontenc}
\usepackage{aecompl}
\usepackage{bbm}
\usepackage[cmex10]{amsmath}
\usepackage[noadjust]{cite}
\usepackage{amssymb}
\usepackage{amsmath}
\usepackage{paralist}
\usepackage{leftidx}
\usepackage{mathrsfs}
\usepackage{multirow}
\usepackage{slashbox}
\usepackage{amsthm}
\usepackage{background}
\usepackage{algorithm}
\usepackage{algorithmic}
\usepackage[norule]{footmisc}
\IEEEoverridecommandlockouts

\theoremstyle{definition}
\newtheorem{defn}{Definition}
\theoremstyle{proposition}
\newtheorem{prop}{Proposition}
\theoremstyle{corollary}

\theoremstyle{theorem}

\theoremstyle{lemma}
\newtheorem{lem}{Lemma}
\newtheorem*{remark}{Remark}
\begin{document}
\SetBgContents{}
\title{Exploiting Mobility in Cache-Assisted D2D Networks: Performance Analysis and Optimization}
\author{\IEEEauthorblockN{Rui Wang, Jun Zhang,~\IEEEmembership{Senior Member,~IEEE}, S.H. Song,~\IEEEmembership{Member,~IEEE}, and Khaled B. Letaief,~\IEEEmembership{Fellow,~IEEE}}
	\thanks{This work was supported by the Hong
		Kong Research Grants Council Grant No. 610113. Part of this work has been presented at IEEE ICC, Paris, France, May 2017 \cite{myICC}.}
	\thanks{R. Wang is with Microsoft, Beijing, P. R. China (email: ruiwa@microsoft.com). This work was done when she was with HKUST. J. Zhang, S.H. Song, and K. B. Letaief are with the Department of Electronic and Computer Engineering, Hong Kong University of Science and Technology, Hong Kong (email: \{eejzhang, eeshsong, eekhaled\}@ust.hk).}}
\maketitle
\vspace{-35pt}
\begin{abstract}
Caching popular content at mobile devices, accompanied by device-to-device (D2D) communications, is one promising technology for effective mobile content delivery. User mobility is an important factor when investigating such networks, which unfortunately was largely ignored in most previous works. Preliminary studies have been carried out, but the effect of mobility on the caching performance has not been fully understood. In this paper, by explicitly considering users' contact and inter-contact durations via an alternating renewal process, we first investigate the effect of mobility with a given cache placement. A tractable expression of the data offloading ratio, i.e., the proportion of requested data that can be delivered via D2D links, is derived, which is proved to be increasing with the user moving speed. The analytical results are then used to develop an effective mobility-aware caching strategy to maximize the data offloading ratio. Simulation results are provided to confirm the accuracy of the analytical results and also validate the effect of user mobility. Performance gains of the proposed mobility-aware caching strategy are demonstrated with both stochastic models and real-life data sets. It is observed that the information of the contact durations is critical to design cache placement, especially when they are relatively short or comparable to the inter-contact durations. 
\end{abstract}
\vspace{-10pt}
\IEEEpeerreviewmaketitle
\newpage

\section{Introduction}

The mobile data traffic is growing at an exponential rate, among which mobile video is dominant and will account for 78\% of the global data traffic by 2021 \cite{cisco}. To accommodate such heavy traffic, network densification is a common solution which can significantly improve the spectral efficiency and network coverage \cite{5G,bhushan2014network}. However, it incurs a heavy burden on the backhaul links. Caching popular contents at helper nodes or user devices is a promising approach to reduce the backhaul data traffic, as well as, improving the user experience for video streaming applications \cite{femtocachingmaga,wang2014cache,cachebenefit,d2d-cache,jcache,maddah2015decentralized}. Exploiting the predictability and reusability of popular content, caching is an effective technology for mobile content delivery, and is the key enabler for content centric wireless networks. In comparison with the commonly considered femto-caching systems \cite{jcache,femtomobility,femtocaching2,femtocachinginfocom,BScache,li2015distributed}
, caching at devices enjoys unique advantages. First, the aggregate cache capacity grows with the number of devices, which will subsequently increase the caching performance \cite{d2d-cache}. Second, device caching can promote device-to-device (D2D) communications, where nearby mobile devices may communicate directly rather than being forced to communicate through the base station (BS), and thus the BS load can be significantly reduced \cite{design}. 

While there have been lots of studies on D2D caching networks \cite{scaling,tradeoff,tradeoffj,fundamentallimits}, an important characteristic of mobile users, i.e., the user mobility, has been largely ignored. Specifically, a fixed topology is normally adopted by assuming users to be at fixed locations, which is not realistic. Recently, a few initial studies on mobility-aware caching have appeared \cite{magmobility,mobilityradom,mobilitycodedcaching,krishnan2017effect,markovmobility,markovmobility_impact,mobilitycaching,scalingmobility}. It has been demonstrated in \cite{mobilitycaching} that mobility-aware D2D caching can help to improve the BS offloading ratio, which, however, was only shown via numerical results. A full understanding of the effect of mobility will require a thorough theoretical analysis, which is not available yet. Moreover, there is some limitation in the mobility model adopted in previous works. For example, it was assumed that a fixed amount of data can be delivered once two users are in contact, while the variation in contact durations was not considered \cite{mobilitycaching,krishnan2017effect,markovmobility_impact,scalingmobility}. As the user mobility will affect both the contact rate and contact duration, it is important to consider their variations when investigating the impact of user mobility, as well as, designing mobility-aware caching strategies. This forms the main objective of this paper.

\subsection{Related Works}
Caching in D2D networks has attracted lots of recent attentions. In \cite{scaling}, the scaling behavior of the number of D2D collaborating links was identified. Three concentration regimes, classified by the concentration of the file popularity, were investigated. The outage-throughput trade-off and optimal scaling laws of both the throughput and outage probability were studied in \cite{tradeoff,tradeoffj}. One main result was that, with a small file library, the throughput is proportional to the ratio of the cache capacity and file library size, while it is independent of the number of users. Two coded caching schemes, i.e., centralized and decentralized, were proposed in \cite{fundamentallimits}, where the contents are delivered via broadcasting. However, a fixed network topology was assumed in most previous works, which is not realistic.

There are some works considering the effect of user mobility, with different mobility models. It was shown through simulations in \cite{mobilitycaching} that a higher user moving speed results in a higher data offloading ratio, while theoretical analysis is missing. In \cite{krishnan2017effect}, a library of two files was considered while each device randomly caches one file. The coverage probability was derived when a user requests the non-cached file and moves from one location to another. Then, it was showed via numerical results that user mobility has a positive effect on D2D caching. Based on a discrete-time Markov process, the impact of user mobility was investigated in \cite{markovmobility_impact}. In this work, several popular locations (e.g., schools and malls) were considered, and it was assumed that users located in the same location are in contact. The throughput-delay scaling law was derived by characterizing the contact rate of the random walk model in \cite{scalingmobility}. However, in these works, it was assumed that the whole caching content or a complete encoded segment can be delivered once two users are in contact, which failed to take the variation in contact durations into account. 

Some preliminary studies also evaluated the effect of user mobility by considering contact durations. In \cite{mobilityradom}, Golrezaei \emph{et al.} validated the performance of their proposed random caching scheme in the mobility scenario via simulations using a random walk model. In \cite{jarray2016effects}, the effect of mobility was evaluated via numerical results. It was assumed that a file can be successfully delivered if a user is in contact with another user caching the requested file, and the contact duration is enough to deliver the whole file. Since only the impact on the contact duration was considered while ignoring the contact rate, it showed that mobility has a negative impact on the hit performance. In \cite{mobilitycodedcaching}, the effect of mobility was evaluated in D2D networks with coded caching, with the conclusion that mobility can improve the scaling law of throughput. In this work, the timeline was divided into discrete time slots, and it was assumed that one coded segment can be delivered in each time slot while two users may keep contact during multiple time slots. However, this result was based on the assumption that the user locations are random and independent in each time slot, which failed to take into account the temporal correlation of user mobility.


\subsection{Contributions}

In this paper, we investigate a D2D caching network with mobile users, by adopting an alternating renewal process to model the mobility pattern so that both the contact and inter-contact durations are accounted for. Specifically, the timeline for an arbitrary pair of mobile users is divided into \emph{contact durations}, which denote the time intervals when the mobile users are located within the transmission range, and \emph{inter-contact durations}, which denote the time intervals between contact durations \cite{pocket}. Meanwhile, both the contact and inter-contact durations are assumed to follow exponential distributions. The \emph{data offloading ratio}, which is defined as the proportion of data that can be obtained via D2D links, is adopted as the performance metric. By theoretically analyzing the data offloading ratio, we first evaluate the effect of user mobility, and then, propose a mobility-aware caching strategy. The main contributions are summarized as follows: 
\begin{itemize}
	\item We derive an accurate expression for the data offloading ratio, for which the main difficulty is to deal with multiple alternating renewal processes. We tackle it by using a beta random variable to approximate the \emph{communication duration} of a given user through moment matching.
	\item We investigate the effect of mobility for a given cache placement. In the low-to-medium mobility scenario, by assuming that the transmission rate does not change with the user speed, it is proved that the data offloading ratio increases with the user speed for any caching strategy that does not cache the same contents at all the users with contacts.
	\item A cache placement problem is formulated in order to maximize the data offloading ratio. By reformulating the original problem into a submodular maximization problem over a matroid constraint, a greedy algorithm is proposed, which can achieve a $\frac{1}{2}$-approximation.
	\item Simulation results validate the accuracy of the derived expression, as well as the effect of user mobility. Moreover, both stochastic models and real-life data sets are used to evaluate the performance of the proposed mobility-aware caching strategy, which is shown to outperform both random and popular caching strategies. Furthermore, it is shown that the variation of contact durations is important while designing caching strategies, especially when the average contact duration is relatively short or comparable to the inter-contact duration.
\end{itemize}
In comparison, our previous work \cite{mobilitycaching} assumed constant contact durations, and did not provide any analytical performance evaluation. With a more realistic mobility model, the cache placement problem formulated in the current paper cannot be solved directly by the algorithm in \cite{mobilitycaching}. Therefore, we improve on both performance analysis and cache placement optimization.

\subsection{Organization}
The remainder of this paper is organized as follows. In Section II, we introduce the mobility and caching models, as well as the performance metric. An approximate expression of the data offloading ratio is derived in Section III. The effect of user mobility is investigated in Section IV, and a mobility-aware caching strategy is proposed in Section V. The simulation results are shown in Section VI. Finally, Section VII concludes the paper.


\section{System Model and Problem Formulation}

In this section, we first introduce the alternating renewal process to model the user mobility pattern, and discuss the caching and file delivery models. Then, the performance metric, i.e., the data offloading ratio, is defined.

\subsection{Mobility Model}
\begin{figure}[!t]
  \centering
  \includegraphics[width=4in]{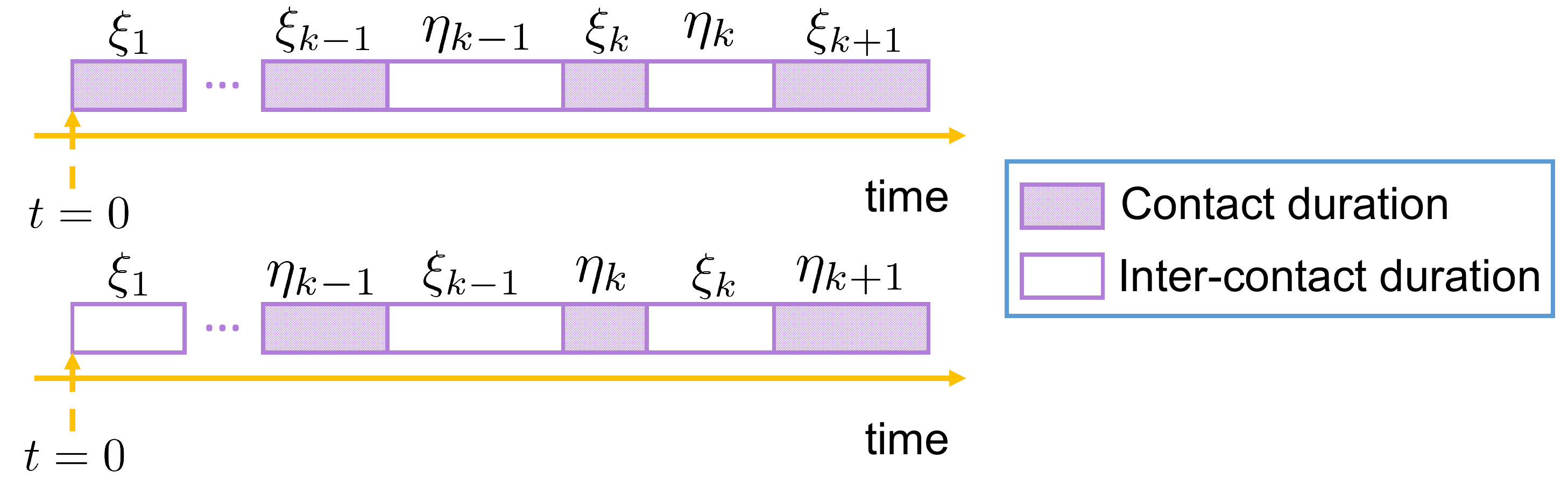}
  \caption{The timeline for an arbitrary pair of mobile users.}
  \label{intercontact}
\end{figure}

The inter-contact model, which captures the temporal correlation of the user mobility \cite{exintercontactmodel}, is adopted to model the user mobility pattern. Specifically, the timeline of each pair of users is divided into  \emph{contact durations}, i.e., the time intervals when the users are in the transmission range, and \emph{inter-contact durations}, i.e., the times intervals between consecutive contact durations. Considering that contact durations and inter-contact durations appear alternatively in the timeline of a user pair, similar to \cite{renewalmodel}, an alternating renewal process \cite{renewalprocess} is applied to model the pairwise contact pattern, as defined below.

\begin{defn}
Consider a stochastic process with state space $\{U,V\}$. The successive durations for the system to be in states $U$ and $V$ are denoted as $\xi_k,k=1,2,\cdots$ and $\eta_k,k=1,2,\cdots$, respectively, which are independent and identically distributed. Specifically, the system starts at state $U$ and remains for $\xi_1$, then switches to state $V$ and stays for $\eta_1$, then backs to state $U$ and stays for $\xi_2$, and so forth. 
Let $\psi_k=\xi_k+\eta_k$. The counting process of $\psi_k$ is called an \emph{alternating renewal process}.
\end{defn}
As shown in Fig. \ref{intercontact}, if a pair of users is in contact at $t=0$, $\xi_k$ and $\eta_k$ represent the contact durations and inter-contact durations, respectively. Otherwise, $\xi_k$ and $\eta_k$ represent the inter-contact durations and contact durations, respectively. It was shown in \cite{exintercontact} that exponential curves well fit the distribution of inter-contact durations, while in \cite{excontact}, it was identified that an exponential distribution is a good approximation for the distribution of the contact durations. Thus, the same as \cite{renewalmodel}, we assume that the contact and inter-contact durations follow independent exponential distributions. For simplicity, the timelines of different user pairs are assumed to be independent. Specifically, we consider a network with $N_u$ users, with the user index set denoted as $\mathcal{S}=\{1,2,\cdots,N_u\}$. The contact and inter-contact durations of users $i  \in \mathcal{S}$ and $j \in \mathcal{S} \backslash \{i\}$ follow independent exponential distributions with parameters $\lambda^C_{i,j}$ and $\lambda^I_{i,j}$, respectively. If users $i  \in \mathcal{S}$ and $j \in \mathcal{S} \backslash \{i\}$ have no contact, the parameters are $\lambda^C_{i,j}=\infty$ and $\lambda^I_{i,j}=0$.
\subsection{Caching and File Transmission Model}
\begin{figure}[!t]
	\centering
	\includegraphics[width=3.5in]{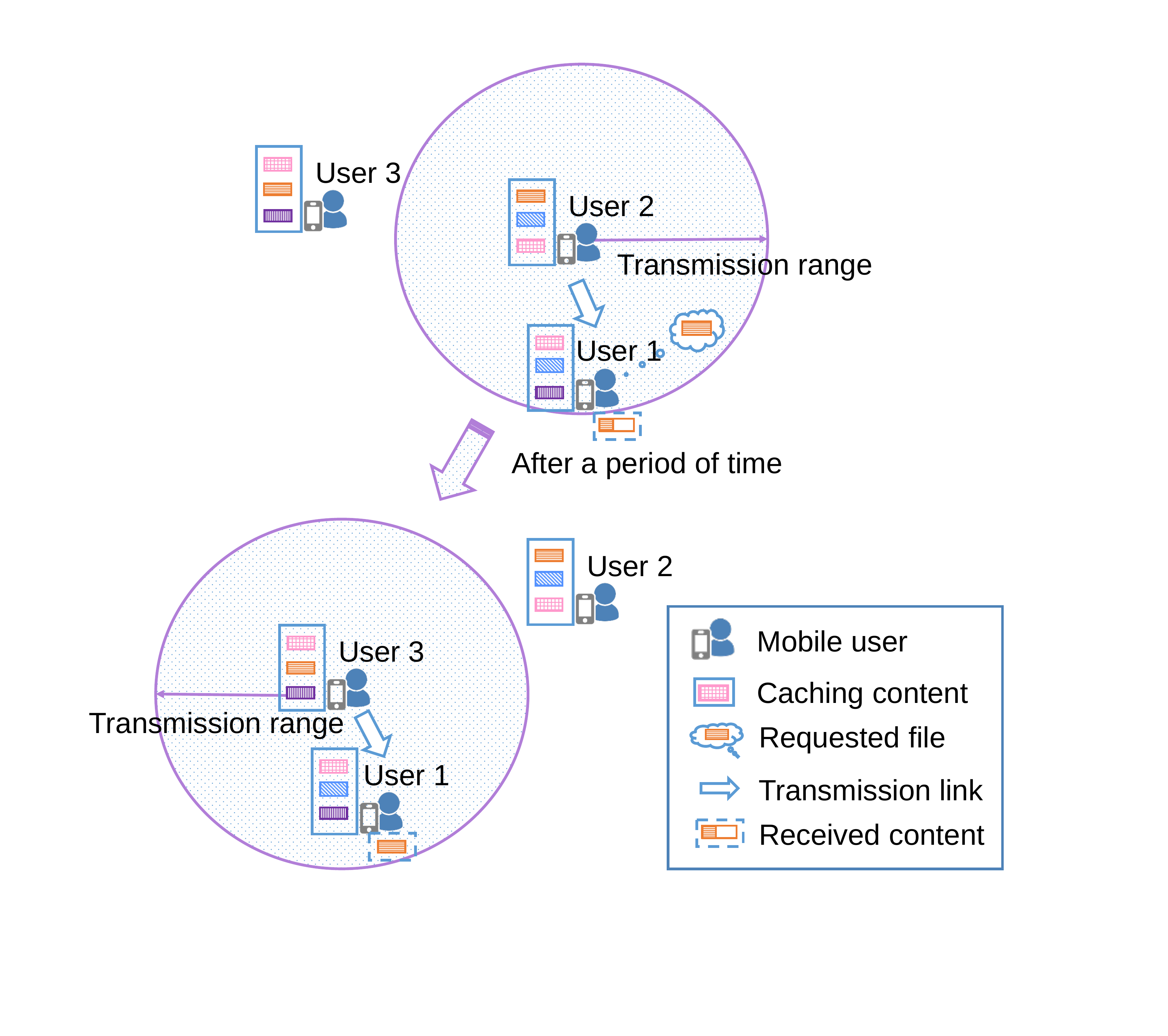}
	\caption{A sample network with three mobile users.}
	\label{model}
\end{figure}

There is a library with $N_f$ files, whose index set is denoted as $\mathcal{F}=\{1,2,\cdots,N_f\}$, each with size $F$. Each user device has a limited storage capacity with size $C$, and each file is completely cached or not cached at all at each user device. Specifically, the cache placement is denoted as
\begin{equation}
x_{j,f}=
\begin{cases}
1, \text{if user $j$ caches file $f$}, \\
0, \text{if user $j$ does not cache file $f$},
\end{cases}
\end{equation} 
where $j \in \mathcal{S}$ and $f \in \mathcal{F}$. User $i \in \mathcal{S}$ is assumed to request a file $f \in \mathcal{F}$ with probability $p^r_{i,f}$, where $\sum \limits_{f \in \mathcal{F} } p^r_{i,f} =1$. When a user requests a file $f$, it will first check its own cache, and then download the file from the users that are in contact and store file $f$, with a fixed transmission rate, denoted as $r_0$. At every time instant, it can only download from one other user. When a user is in contact with multiple users caching the requested file at the same time, it will randomly choose one to download. Meanwhile, D2D pairs are under the control of the base station, and orthogonal resource allocation is assumed for different simultaneous D2D communication pairs. Thus, there is no inter-user interference. We also assume that each user only requests one file at each time, and a new request is generated after it finishes downloading the previous file. If the user cannot get the whole file within a certain delay threshold, denoted as $\tau_0$, it will download the remaining part from the BS. We assume that the delay threshold is larger than the time required to download each content (i.e., $\tau_0>\frac{F}{r_0}$). Fig. \ref{model} shows a sample network, where user $1$ gets part of the requested file during the contact with user $2$, then gets the whole file after encountering user $3$. 

\subsection{Performance Metric}
\begin{figure}[!t]
	\centering
	\includegraphics[width=4in]{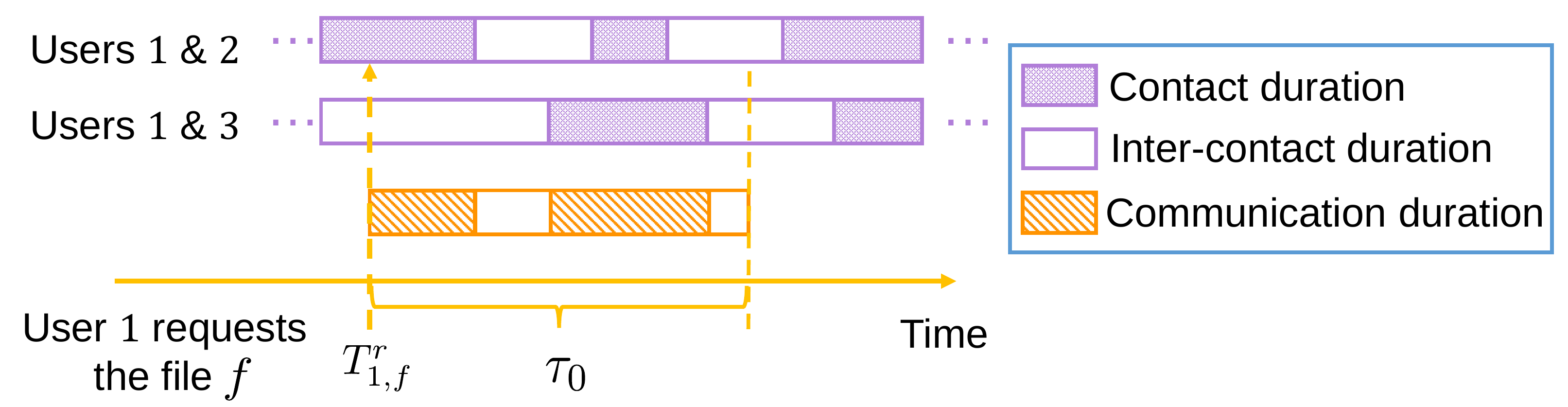}
	\caption{An illustration of the communication duration.}
	\label{transmissiontime}
\end{figure}
The \emph{data offloading ratio}, which is defined as the expected percentage of requested content that can be obtained via D2D links rather than downloading from the BS, is used as the performance metric. Specifically, the data offloading ratio for user $i \in \mathcal{S}$ requesting file $f \in \mathcal{F}$, with $x_{i,f}=0$, is defined as
\begin{equation}
\mathcal{R}_{i,f}=\mathbb{E} _{D_{i,f}}\left[ \min \left( D_{i,f}/F ,1 \right) \right],
\end{equation}
where $\mathbb{E}[\cdot]$ denotes the expectation and $D_{i,f}$ denotes the amount of requested data that can be delivered via D2D links when user $i$ requests file $f$. Since a fixed transmission rate is assumed, $D_{i,f}$ can be written as $D_{i,f}=r_0 \tau^c_{i,f}$, where $\tau^c_{i,f}$ is the \emph{communication duration} for user $i$ to download file $f$ from other users caching file $f$ within time $\tau_0$. We assume that user $i$ can download file $f$ while at least one user caching file $f$ is in contact, where the handover time is ignored. Fig. \ref{transmissiontime} shows the communication duration of user $1$ for the situation in Fig. \ref{model}, where $T_{1,f}^r$ is the time instant when user $1$ requests a file $f \in \mathcal{F}$. Thus, we get
\begin{equation} \label{ifratio}
\mathcal{R}_{i,f}=\mathbb{E} _{\tau^c_{i,f}}\left[ \min \left( r_0 \tau^c_{i,f}/F ,1 \right) \right].
\end{equation}
Then, the data offloading ratio when user $i$ requests file $f$ is $\left[ x_{i,f} + (1- x_{i,f}) \mathcal{R}_{i,f}) \right)]$. Considering the file request probabilities, the overall data offloading ratio is
\begin{align} \label{define_ratio}
\mathcal{R}= \frac{1}{N_u}\sum \limits_{i \in \mathcal{S}} \sum \limits_{f \in \mathcal{F}} p^r_{i,f} \left[ x_{i,f} +  (1-x_{i,f}) \mathcal{R}_{i,f} \right].
\end{align}
While this performance metric has been used in \cite{mobilitycaching} to design mobility-aware caching strategies, the variation of contact durations has not been considered. Furthermore, there is no theoretical analysis available, and it is unclear how user mobility will affect the caching performance. We shall first provide theoretical analysis in Section III for the data offloading ratio with a given cache placement. The analytical result will then be used in Section IV to reveal the effect of user moving speeds, and in Section V to optimize the caching strategy.


\section{Theoretical Analysis of Data Offloading Ratio}
The main difficulty in theoretically evaluating the data offloading ratio is to get the distribution of the communication duration. As this distribution is highly complicated, instead of deriving it directly, we develop an accurate approximation. In this section, we first approximate the distribution of the communication duration by a beta distribution, and then an approximation of the data offloading ratio is obtained.

\subsection{Communication Duration Analysis}
To assist the analysis of the communication duration, we first define some stochastic processes.
\begin{defn} \label{hij}
We model the timeline of a pair of users $i,j \in S, i\neq j$ as an alternating renewal process, and denote $H_{i,j}(t)$, $t \in [0,\infty)$ as the state at time $t$. Specifically, $H_{i,j}(t)=1$ means that users $i$ and $j$ are in contact at time instant $t$. Otherwise, $H_{i,j}(t)=0$. The durations of staying in states $1$ and $0$ follow independent exponential distributions with parameters $\lambda^C_{i,j}$ and $\lambda^I_{i,j}$, respectively.  
\end{defn}
\begin{defn} \label{hif}
When requesting file $f \in \mathcal{F}$, the timeline of user $i \in \mathcal{S}$ can be divided into intervals when it is in contact with at least one user caching file $f$, and intervals that user $i$ is out of the transmission range of all the users caching file $f$. We use a random process to model the timeline of user $i$ requesting file $f$, and denote $H^f_{i}(t)$, $t \in [0,\infty)$ as the state at time $t$. Specifically, $H^f_{i}(t)=1$ means that user $i$ can download file $f$ from a user caching file $f$ at time instant $t$. Otherwise, $H^f_{i}(t)=0$.
\end{defn}
Based on Definitions \ref{hij} and \ref{hif}, we can get the relationship between $H^f_{i}(t)$ and $H_{i,j}(t)$ as $H^f_i(t)=1- \prod \limits_{j \in \mathcal{S} \backslash \{i\},x_{j,f}=1}\left[1-H_{i,j}(t) \right]$. Similar to \cite{renewalmodel}, it is reasonable to assume that when a user requests a file, the alternating process between each pair of users has been running for a long time. Thus, denote $T^r_{i,f}$, $i \in \mathcal{S}$ and $f \in \mathcal{F}$, as the time instant when user $i$ requests file $f$, and the corresponding communication duration $\tau^c_{i,f}$ can be derived as
\begin{equation} \label{time}
\tau^c_{i,f}= \lim \limits_{T^r_{i,f} \to \infty} \int _{T^r_{i,f}}^{T^r_{i,f}+\tau_0} H^f_i(t) dt.
\end{equation}
In the following lemma, we derive the expectation and variance of the communication duration.

\begin{lem} \label{ev}
When user $i \in \mathcal{S}$ requests file $f \in \mathcal{F}$, with $x_{i,f}=0$, the expectation and variance of its communication duration are
\begin{equation} \label{expectation_t}
\mathbb{E}\left[\tau^c_{i,f} \right]=\tau_0\left( 1-\prod \limits_{j \in \mathcal{S}, x_{j,f}=1} p_{i,j}^I\right),
\end{equation}
and
\begin{align} \label{var_t}
\mathrm{Var} \left[ \tau^c_{i,f} \right] = & 2 \int_0^{\tau_0} (\tau_0-u) \prod \limits_{j \in \mathcal{S}, x_{j,f}=1} p_{i,j}^I \left[ p_{i,j}^I + \left( 1-p_{i,j}^I\right) e^{-u(\lambda_{i,j}^C+\lambda^I_{i,j})} \right] du \notag \\
& -\tau_0^2 \prod \limits_{j \in \mathcal{S}, x_{j,f}=1} \left(p_{i,j}^I\right)^2,
\end{align}
where $p_{i,j}^I=\frac{\lambda^C_{i,j}}{\lambda^C_{i,j}+\lambda^I_{i,j}}$ denotes the probability that users $i$ and $j$ are not in contact at a time instant $t \to \infty$.
\end{lem}

\begin{proof}
See Appendix A.
\end{proof}
While the mean and variance have been derived, the exact distribution of the communication duration is still highly intractable. Since it is a bounded random variable, we propose to approximate the distribution of $\tau^c_{i,f}$ by a beta distribution, which has been widely used to model random variables limited to finite ranges. For example, it has been used to model the prior distribution for a Bernoulli trial and the timestamp over a time span \cite{betaexample_time}. The pdf of a beta distribution is 
\begin{equation}
f(x) = \frac{x^{\alpha-1} (1-x)^{\beta-1}}{B(\alpha,\beta)},
\end{equation}
where $\alpha$ and $\beta$ are shape parameters, and $B(\cdot,\cdot)$ is the beta function. Specifically, if $\sum_{j \in \mathcal{S} \backslash \{i\}} x_{j,f}>0$, which means that user $i$ may download file $f$ from at least one other user, we approximate $\tau^c_{i,f}/{\tau_0}$ by $Z_{i,f}$, where $Z_{i,f} \sim \text{Beta}(\alpha_{i,f},\beta_{i,f})$, $i \in \mathcal{S}$ and $f \in \mathcal{F}$. Otherwise, $\tau^c_{i,f}=0$. By matching the first two moments (i.e.,  $\mathbb{E}[Z_{i,f}]=\mathbb{E}[\tau^c_{i,f}/\tau_0]$ and $\mathrm{Var}[Z_{i,f}]=\mathrm{Var}[\tau^c_{i,f}/\tau_0]$) the parameters of the beta distribution can be derived as\footnote{The parameters of the beta distribution should be positive, and it can be proved that $\alpha_{i,f}>0$ and $\beta_{i,f}>0$, by $e^{-u(\lambda_{i,j}^I+\lambda_{i,j}^C)} \le 1$.}
\begin{equation} \label{beta_p}
\begin{cases}
\alpha_{i,f}=\frac{\mathbb{E}[\tau^c_{i,f}]^2 (\tau_0-\mathbb{E}[\tau^c_{i,f}])}{ \mathrm{Var}[\tau^c_{i,f}]\tau_0}-\frac{\mathbb{E}[\tau^c_{i,f}]}{\tau_0}, \\
\beta_{i,f}=\frac{\tau_0-\mathbb{E}[\tau^c_{i,f}]}{\mathbb{E}[\tau^c_{i,f}]} \alpha_{i,f}.
\end{cases}
\end{equation}

\subsection{Data Offloading Ratio Analysis}  
Based on the above approximation, we get an approximate expression of the data offloading ratio in Proposition \ref{E_od}.
\begin{prop} \label{E_od}
The data offloading ratio is
\begin{align} \label{ratio}
\mathcal{R} = &\frac{1}{N_u}\sum \limits_{i \in \mathcal{S}} \sum \limits_{f \in \mathcal{F}} p^r_{i,f} \left[ x_{i,f} + (1-x_{i,f}) \mathcal{R}_{i,f} \right],
\end{align}
where $\mathcal{R}_{i,f}$ is approximated by
\begin{align} \label{ratio_g}
\mathcal{R}^a_{i,f} = 1-I_{\frac{F}{\tau_0 r_0}} \left(\alpha_{i,f},\frac{\tau_0-\mathbb{E}[\tau^c_{i,f}]}{\mathbb{E}[\tau^c_{i,f}]} \alpha_{i,f}\right)
+\frac{\mathbb{E}[\tau^c_{i,f}]r_0}{F} I_{\frac{F}{\tau_0 r_0}}\left(\alpha_{i,f}+1,\frac{\tau_0-\mathbb{E}[\tau^c_{i,f}]}{\mathbb{E}[\tau^c_{i,f}]} \alpha_{i,f}\right) ,
\end{align}
if $\sum_{j \in \mathcal{S} \backslash \{i\}} x_{j,f}>0$, and $\mathcal{R}_{i,f}^a=0$ elsewhere, where $I_q(\cdot,\cdot)$ is the incomplete beta function, $\mathbb{E}[\tau^c_{i,f}]$ is given in (\ref{expectation_t}), and $\alpha_{i,f}$ is given in (\ref{beta_p}).
\end{prop}

\begin{proof}
See Appendix B.
\end{proof}
\begin{remark}
	To evaluate the data offloading ratio in Proposition \ref{E_od}, the main complexity is in calculating the integral in (\ref{var_t}) when calculating the variance of the communication duration. Simulations will show that the approximation is quite accurate. This analytical result will serve as the basis for evaluating the effect of user mobility in Section IV and designing the mobility-aware caching strategy in Section V.
\end{remark}

\section{Effect of User Moving Speed}

In this section, we investigate how changing the user moving speed will affect the data offloading ratio for a given caching strategy. If all the user pairs $i,j \in \mathcal{S}, i \neq j$ with contacts (i.e., $0 < \lambda_{i,j}^I,\lambda_{i,j}^C  < \infty$) cache the same contents, D2D communications will not help the content delivery. Thus, in the following, we assume that the contents cached at the users with contacts are not all the same. In other words, there exists a pair of users $i,j \in \mathcal{S}, i\neq j$ with $0 < \lambda_{i,j}^I,\lambda_{i,j}^C  < \infty$, and a file $f \in \mathcal{F}$, such that $x_{i,f}=1$ and $x_{j,f}=0$. This investigation is based on the approximate expression in (\ref{ratio}), and simulations will be provided later to verify the results.
\subsection{Effect of Moving Speed on Communication Duration}

We first characterize the relationship between the user moving speed and the parameters $\lambda^C_{i,j}$ and $\lambda^I_{i,j}$ in Lemma \ref{speed}.
\begin{lem} \label{speed}
When all the user speeds change by $\mu$ times, the contact and inter-contact parameters will also change by $\mu$ times (i.e., from $\lambda_{i,j}^C$ and $\lambda_{i,j}^I$ to $\mu\lambda_{i,j}^C$ and $\mu\lambda_{i,j}^I$, $i,j \in \mathcal{S}$, respectively).
\end{lem}
\begin{proof}
The time for user $i$ to move along a certain path $L_i$ can be given as a curve integral $\int_{L_i} \frac{dz}{v_i(z)}$, where $v_i(z)$ is the speed of user $i$ when passing by a point $z$ on the path $L_i$. When the speed of user $i$ changes by $\mu$ times, the time for moving along the path $L_i$ changes to $\int_{L_i} \frac{dz}{\mu v_i(z)}=\frac{1}{\mu}\int_{L_i} \frac{dz}{ v_i(z)}$, which is scaled by $\frac{1}{\mu}$ times. In each contact or inter-contact duration, users $i$ and $j$ move along certain paths. When all the user speeds change by $\mu$ times, each contact or inter-contact duration changes by $\frac{1}{\mu}$ times, and thus, the average values also change by $\frac{1}{\mu}$ times. Since the contact and inter-contact durations are assumed to be exponentially distributed with means $\frac{1}{\lambda_{i,j}^C}$ and $\frac{1}{\lambda_{i,j}^I}$, respectively, the parameters $\lambda_{i,j}^C$ and $\lambda_{i,j}^I$ are scaled by $\mu$ times.
\end{proof}
Considering that a larger $\mu$ means that users are moving faster, in the following, we will investigate how changing  $\mu$ will affect the data offloading ratio. For simplicity, we assume that the transmission rate is a constant, and will not change with the user speed. This is reasonable in the low-to-medium mobility regime. 
Firstly, when the user speed changes by $\mu$ times, we rewrite the expectation and variance of communication duration as in Lemma \ref{evv}.
\begin{lem} \label{evv}
	When the user speed changes by $\mu$ times, the expectation of the communication duration when user $i \in \mathcal{S}$ requests file $f \in \mathcal{F}$ is the same as (\ref{expectation_t}), and the corresponding variance is given as	
	\begin{equation} \label{v_ts}
	\mathrm{Var}[\tau^c_{i,f}(\mu)]= 2 a_0 \sum \limits_{\mathbf{Z} \in \{0,x_{l,f}\}^{N_f} \backslash \mathbf{0}} \frac{a_{\mathbf{Z}}}{\mu\kappa_{\mathbf{Z}}} \left( \tau_0 -\frac{1}{\mu\kappa_{\mathbf{Z}}} +\frac{1}{\mu\kappa_{\mathbf{Z}}} 
	\exp(-\mu\kappa_{\mathbf{Z}}\tau_0 ) \right),
	\end{equation}
	where $\mathbf{Z}$ is an $N_f$-ary vector, with $z_l$ as the $l$-th element of $\mathbf{Z}$, and
	\begin{equation} \label{ax}
	\begin{cases}
	a_0=\prod \limits_{j \in \mathcal{S}, x_{j,f}=1} \left(p_{i,j}^I \right)^2, \\
	a_{\mathbf{Z}}=\prod \limits_{l=1}^{N_f} \left(\frac{\lambda_{i,l}^I}{\lambda_{i,l}^C} \right)^{z_l}, \\
	\kappa_{\mathbf{Z}}=\sum \limits_{l=1}^{N_f} z_l (\lambda_{i,l}^C+\lambda_{i,l}^I).
	\end{cases}
	\end{equation}
\end{lem}
\begin{proof}
	See Appendix C.
\end{proof}
\noindent Then, the effect of user speed on the communication duration is shown in Lemma \ref{s_t} .
\begin{lem} \label{s_t}
	 When $\mu$ increases, the expectation of the communication duration for user $i \in \mathcal{S}$ requesting file $f \in \mathcal{F}$ (i.e., $\mathbb{E}[\tau^c_{i,f}(\mu)]$) does not change, and the corresponding variance (i.e., $\mathrm{Var}[\tau^c_{i,f}(\mu)]$) decreases, if there is at least one user caching file $f$ has contacts with user $i$. Accordingly, the parameter $\alpha_{i,f}$ of the beta distribution increases.
\end{lem}
\begin{proof}
	See Appendix D.
\end{proof}
\begin{remark}
	Intuitively, when the users move faster, they have more contacts with each other within a certain time period $\tau_0 $, while the duration of each contact decreases. That is why the variance of communication durations decreases and the expected value does not change.
\end{remark}

\subsection{Effect of Moving Speed on Data Offloading Ratio}
In the following, we evaluate the relationship between $\alpha_{i,f}$ and the data offloading ratio when user $i$ requests file $f$ that is not in its own cache (i.e., $\mathcal{R}_{i,f}$ in (\ref{ratio_g})) in Lemma \ref{s_g}.
\begin{lem} \label{s_g}
	When user $i \in \mathcal{S}$ requests file $f \in \mathcal{F}$ and $x_{i,f}=0$, the data offloading ratio (i.e., $\mathcal{R}_{i,f}$) increases with parameter $\alpha_{i,f}$.
\end{lem}
\begin{proof}
	See Appendix E.
\end{proof}
\noindent Based on Lemmas \ref{s_t} and \ref{s_g}, we specify the effect of the user speed in Proposition \ref{s_d}.
\begin{prop}\label{s_d}
	If the transmission rate keeps unchanged, the data offloading ratio increases with the user moving speed.
\end{prop}
\begin{proof}
	See Appendix F.
\end{proof}
\begin{remark}
	The result in Proposition \ref{s_d} is valid for any caching strategy, only excluding the case that all the user pairs with contacts have the same cache contents. Although a higher mobility shortens the contact durations, it provides more opportunities for the users to share data, and improves the overall data offloading ratio. This can be regarded as a type of \emph{diversity gain}. Thus, it is important to take advantage of user mobility in the cache-assisted D2D network.
\end{remark}

\section{Mobility-Aware Caching Strategy}
We have derived the expression for the data offloading ratio and theoretically evaluated the effect of user mobility. These analytical results are also useful for designing mobility-aware caching strategies. Recently, some works started to develop cache placement strategies for D2D networks by taking advantage of user mobility information \cite{mobilitycaching,markovmobility}. However, it was assumed that a complete file or segment can be delivered when two users are in contact, and the variation in the contact duration was ignored. In this section, we formulate a cache placement problem and propose a mobility-aware caching strategy using both the contact and inter-contact information.

\subsection{Problem Formulation}

We will develop a caching strategy, which takes advantage of both the contact and inter-contact information to improve the data offloading ratio. With a higher data offloading ratio, more D2D links can be established, and thus, the spectrum efficiency can be improved. Specifically, we consider that the contact and inter-contact parameters can be estimated via historical data. The mobility-aware cache placement problem is formulated as
\begin{align}
\max \limits_{x_{i,f}} \quad & \mathcal{R}, \label{problem} \\
\text{s.t.} \quad & F \sum \limits_{f \in \mathcal{F}} x_{i,f} \leqslant C, i \in \mathcal{S}, \tag{\ref{problem}a} \\
& x_{i,f} \in \{0,1\}, i \in \mathcal{S}, \tag{\ref{problem}b}
\end{align}
where constraint (\ref{problem}a) implies a limited cache capacity and constraint (\ref{problem}b) means that each file is fully cached or not cached at all.

\subsection{Submodular Functions and Matroid Constraints}
Our proposed algorithm is based on submodular optimization. As a typical kind of combinational optimization problems, submodular maximization has been extensively investigated \cite{submodularbook,submodular2,submodular}. It has a wide range of applications, including the max-$k$-cover problem and the max-cut problem \cite{1e}. Moreover, in wireless networks, submodular maximization over a matroid constraint has been utilized to design network coding \cite{networkcoding} and femto-caching \cite{femtocaching2}. For completeness, we will present some necessary definitions and properties related to the submodular set function and the matroid constraint. Please refer to \cite{submodular2} for more details.
\begin{defn}
	Let $S$ be a finite ground set, and a function $h: 2^{S} \to \mathbb{R}$ is a \emph{submodular set function} if $h(A)+h(B) \geqslant h(A \cup B) + h(A \cap B)$, $\forall A \subseteq S \text{ and } \forall B \subseteq S$.
\end{defn}
\begin{prop} \label{secondorder}
	Let $A \subset S \text{ and } j,k \in S-A, j \ne k$, and a function $h: 2^{S} \to \mathbb{R}$ is a monotone submodular function if
	$h(A \cup \{k\}) - h(A) \geqslant h(A \cup \{j,k\}) - h(A \cup\{j\}) \geqslant 0.$
\end{prop}
\noindent Proposition \ref{secondorder} provides an intuitive explanation of the submodular property, i.e.,  the marginal gain of a monotone submodular set function decreases with a larger set. It is also useful when proving submodularity. 
\begin{defn}
	A pair $\mathcal{M}=\{S,\mathcal{I}\}$, where $S$ is a finite ground set and $\mathcal{I}$ is a collection of subsets of $S$,  is called a \emph{matroid} if
	\begin{enumerate}
		\item $\varnothing \in \mathcal{I}$,
		\item If $A \subseteq B \subseteq S$ and $ B \in \mathcal{I}$, then $A \in \mathcal{I}$,
		\item If $A,B \in \mathcal{I}$ and $|B|>|A|$, then $\exists j \in B-A$ such that $A \cup \{j\} \in \mathcal{I}$.
	\end{enumerate}
\end{defn}

\subsection{Problem Reformulation}
In the following, problem (\ref{problem}) will be reformulated as a submodular maximization problem over a matroid constraint. We first define the ground set as $S=\{y_{j,f}|j \in \mathcal{S} \text{ and } f \in \mathcal{F} \}$, and a cache placement can be represented as a subset of $S$. Specifically, for a cache placement $Y \subseteq S$, $y_{j,f} \in Y$ means user $j$ caches file $f$ (i.e., $x_{j,f}=1$), and $y_{j,f} \notin Y$ means user $j$ does not cache file $f$ (i.e., $x_{j,f}=0$). Accordingly, the data offloading ratio can be rewritten as 
\begin{align} \label{ratio_Y}
\mathcal{R}(Y) = \frac{1}{N_u}\sum \limits_{i \in \mathcal{S}} \sum \limits_{f \in \mathcal{F}} p^r_{i,f} \left[ \mathbbm{1}(y_{i,f} \in Y) + \mathbbm{1}(y_{i,f} \notin Y) \mathcal{R}_{i,f}(Y) \right],
\end{align}
where 
\begin{equation} \label{rifY}
\mathcal{R}_{i,f}(Y)=\mathbb{E}_{\tau^c_{i,f}}[\min(r_0 \tau^c_{i,f}(Y))/F,1].
\end{equation}
Lemma \ref{sub_p} shows that $\mathcal{R}(Y)$ is a monotone submodular function by verifying Proposition \ref{secondorder}.
\begin{lem} \label{sub_p}
	$\mathcal{R}(Y)$ in (\ref{ratio_Y}) is a monotone submodular set function on the ground set $S$.
\end{lem}
\begin{proof}
	See Appendix E.
\end{proof}
\noindent Then, Lemma \ref{mat_p} rewrites the constraint in problem (\ref{problem}) as a matroid constraint.
\begin{lem} \label{mat_p}
	Let $S_i=\{y_{i,f}| f \in \mathcal{F}\}$, include all files that may be cached at user $i$, the constraint in problem (\ref{problem}) is equivalent to a matroid constraint $Y \in \mathcal{I}$, where
	\begin{equation} \label{motroid_I}
	\mathcal{I}=\left\{Y \subseteq S \Big\vert |Y \cap S_i| \leqslant C/F, \forall i \in \mathcal{S} \right\}.
	\end{equation}
\end{lem}
\begin{proof}
	Constraint (\ref{motroid_I}) is a partition matroid, which is a typical matroid \cite{submodular2}. 
\end{proof}
\noindent Thus, problem (\ref{problem}) can be reformulated as the following monotone submodular maximization problem over a matroid constraint.
\begin{equation}
\max \limits_{Y \in \mathcal{I}} \quad \mathcal{R}(Y) = \frac{1}{N_u}\sum \limits_{i \in \mathcal{S}} \sum \limits_{f \in \mathcal{F}} p^r_{i,f} \left[ \mathbbm{1}(y_{i,f} \in Y) + \mathbbm{1}(y_{i,f} \notin Y) \mathcal{R}_{i,f}(Y) \right].
\end{equation}

\subsection{Greedy Cache Placement Algorithm}
To solve a monotone submodular maximization problem, a greedy algorithm provides a $\frac{1}{2}$-approximation \cite{submodular2}, which means that the solution is at least $50\%$ of the optimal one. Although a randomized algorithm can get a higher approximation ratio as $\left( 1-1/e \right)$ \cite{1e}, its computational complexity is too high as the size of the ground set is $N_u \times N_{f}$. Thus, it is inapplicable in practice. Moreover, it will be shown in the simulation that the greedy algorithm provides a near optimal performance.

Denote $S^r$ as the set including all the elements that can be added into the cache placement set $Y$. We define the priority value of element $y_{j,f} \in S^r$ as the gain of adding $y_{j,f}$ into $Y$, which is given as
\begin{align}
g_{j,f}=&\mathcal{R}(Y \cup\{y_{j,f}\})-\mathcal{R}(Y) \notag \\
= & \frac{1}{N_u}\left\{ p^r_{j,f}\left[1-\mathcal{R}_{j,f}(Y)\right] + \sum \limits_{i \in \mathcal{S}\backslash \{j\}} p^r_{i,f} 
\mathbbm{1}(y_{i,f} \notin Y) \left[ \mathcal{R}_{i,f}(Y \cup \{y_{j,f}\})- \mathcal{R}_{i,f}(Y) \right] \right\}.
\end{align}
The main difficulty when developing the greedy algorithm is to efficiently evaluate the priority values, where the key is to evaluate the data offloading ratio when user $i$ requests file $f$ (i.e., $\mathcal{R}_{i,f}(Y)$). In Section III, we have provided an approximation for $\mathcal{R}_{i,f}(Y)$ as (\ref{ratio_g}), which can be used to evaluate the priority values. The procedure of the greedy algorithm is listed in Algorithm \ref{alg:greedy}. It starts from an empty set $Y= \varnothing$, which means that no file is cached. At each iteration, an element $y_{j^\star,f^\star} \in S^r$ with the maximum priority value is selected and added into $Y$. Meanwhile, $y_{j^\star,f^\star}$ is excluded from $S^r$. If user $j^\star$ cannot cache more files, all the elements in $S_{j^\star}$ are excluded from $S^r$. The process continues until no more file can be cached. 
\begin{algorithm}[!h]
	\caption{The Greedy Cache Placement Algorithm}
	\label{alg:greedy}
	\begin{algorithmic}[1]
		\STATE {Set $Y=\varnothing \Leftrightarrow \text{ set }x_{j,f}=0, \forall j \in \mathcal{S} \text{ and } f \in \mathcal{F}$.}
		\STATE {$S^r=S$.}
		\STATE {Initialize the priority values $
			\left\{ g_{j,f}=\mathcal{R}(Y \cup \{y_{j,f}\})-\mathcal{R}(Y)| j \in \mathcal{S} \text{ and } f \in \mathcal{F} \right\}$.}
		\WHILE{$|Y| < \frac{N_u C}{F} $}
		\STATE {$y_{j^\star,f^\star}=\arg \max \limits_{y_{j,f} \in S^r}  g_{j,f}$.}
		\STATE {Set $Y=Y \cup \{y_{j^\star,f^\star}\} \Leftrightarrow \text{ set } x_{j^\star,f^\star}=1$.}
		\STATE {$S^r=S^r-\{y_{j^\star,f^\star}\}$.}
		\IF {$|Y \cap S_{j^\star}|+1 > C/F$}
		\STATE {$S^r=S^r-S_{j^\star}$}
		\ENDIF
		\STATE {Update the priority values $\{ g_{j,f^\star} | j \in \mathcal{S}, y_{j,f^\star} \in S^r \}$.}
		\ENDWHILE
	\end{algorithmic}
\end{algorithm}

\subsection{Computational Complexity}
There are in total $N_uC/F$ iterations in Algorithm \ref{alg:greedy}. At each iteration, the complexity of updating the priority value is $\mathcal{O}(N_u T_v)$, where $T_v$ is the complexity of calculating the variance in (\ref{var_t}), and the complexity to find the maximum priority value is $\mathcal{O}(N_f)$. Thus, the overall computational complexity of Algorithm \ref{alg:greedy} is $\mathcal{O}\left(N_uC/F \left( N_u T_V+N_f \right) \right)$.
   
\section{Simulation Results}
In the simulation, we first validate the analytical results in Sections III and IV, and then, evaluate the performance of the proposed mobility-aware caching strategy. The content request probability is assumed to follow a Zipf distribution with parameter $\gamma_r$ (i.e., $p^r_{i,f}=\frac{f^{-\gamma_r}}{\sum \limits_{l \in \mathcal{F}} l^{-\gamma_r}}$, $i \in \mathcal{S}$ and $f \in \mathcal{F}$) \cite{d2d-cache}. 

\subsection{Data Offloading Ratio}

In this part, we evaluate the approximate expression of the data offloading ratio in (\ref{ratio}) via simulations. A random caching strategy \cite{randomcache} is applied, where the probabilities of the contents cached at each user are proportional to the file request probabilities. The inter-contact parameters $\lambda^I_{i,j}, i \in \mathcal{S}, j \in \mathcal{S} \backslash \{i\}$ are generated according to a gamma distribution as $\Gamma(4.43,1/1088)$ \cite{aggregate_real}. Similar as \cite{renewalmodel}, we assume the average inter-contact durations to be $5$ times as large as the average contact durations. Thus, the contact parameters are generated according to $\Gamma(4.43 \times 25,1/1088/5)$. The theoretical results are obtained by (\ref{ratio}), and the simulation results are obtained by randomly generating the contact and inter-contact durations according to exponential distributions.

\begin{figure}
	\begin{minipage}[t]{0.49\linewidth}
		\centering
		\includegraphics[width=3in]{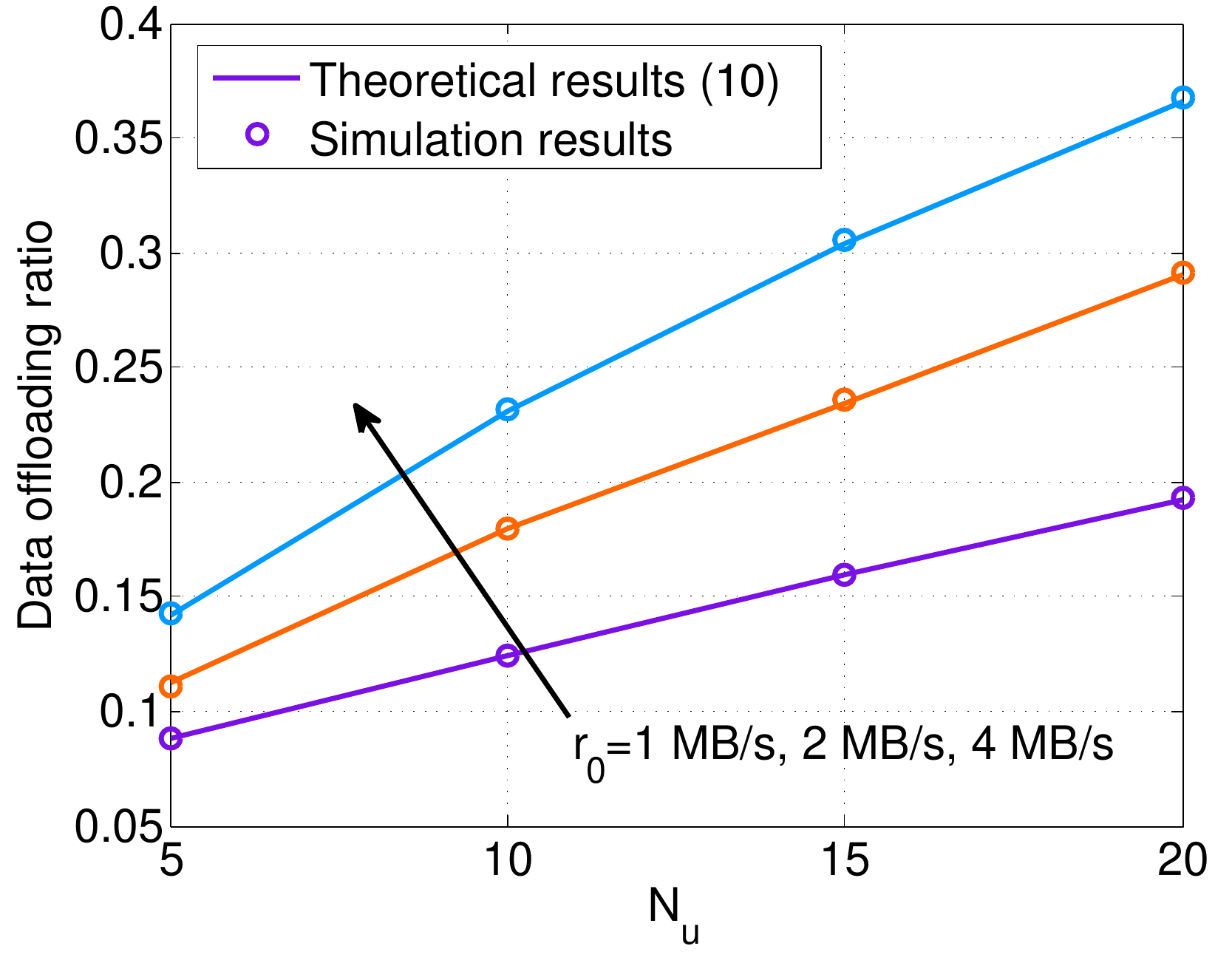}
		\caption{Data offloading ratio with $N_f=100$, $\tau_0 =300$ s, $\gamma_r=0.6$, $C=1$ GB, $F=300$ MB.}
		\label{fig_scheme1-1}
	\end{minipage}%
	\hspace{0.1in}
	\begin{minipage}[t]{0.49\linewidth}
		\centering
		\includegraphics[width=3in]{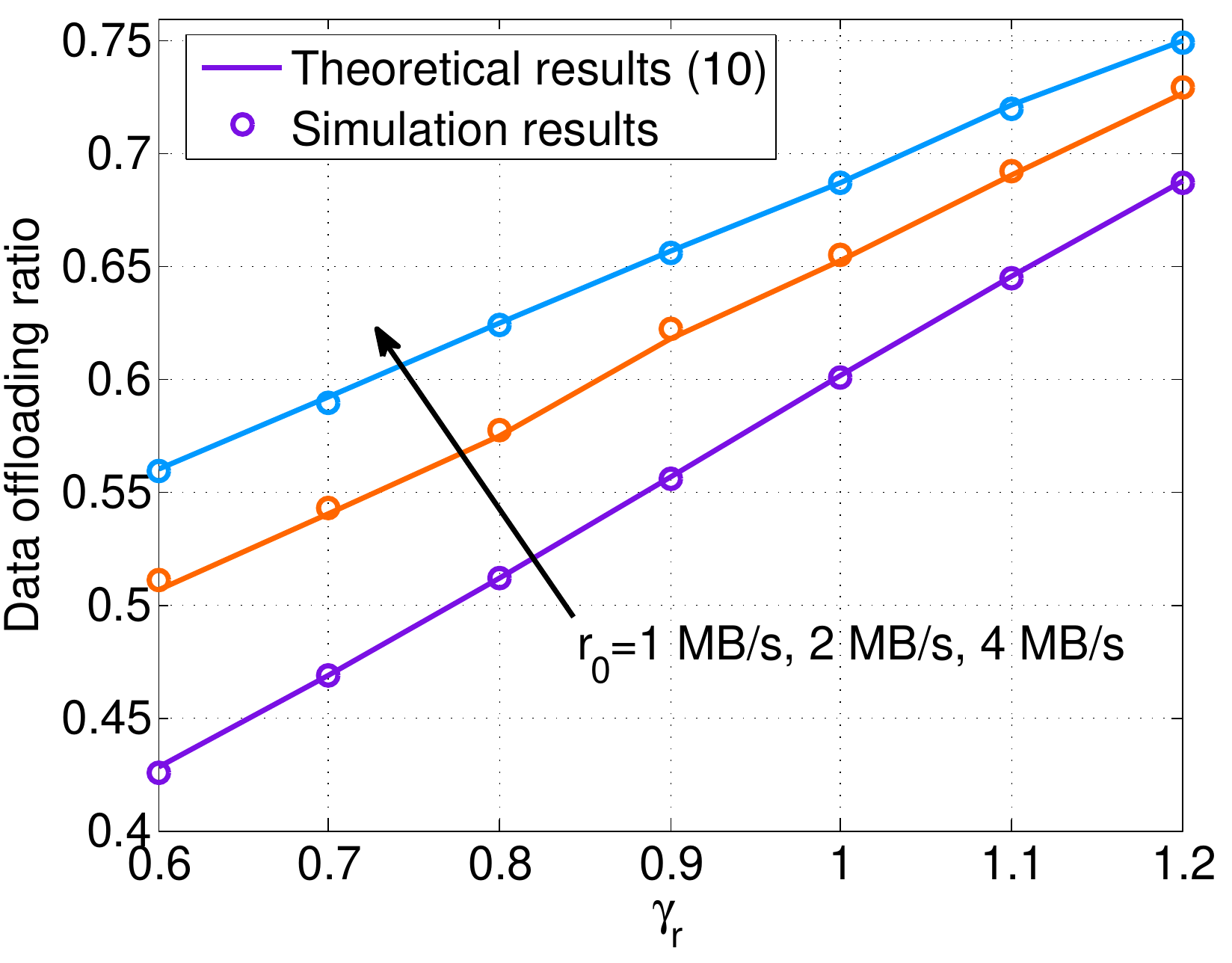}
		\caption{Data offloading ratio with $N_u=15$, $N_f=100$, $\tau_0 =120$ s, $C=1$ GB, $F=100$ MB.}
		\label{fig_scheme1-2}
	\end{minipage}
\end{figure}

Fig. \ref{fig_scheme1-1} validates the accuracy of the approximation in (\ref{ratio}) by varying the number of users. It is shown that the theoretical results are very close to the simulation results, which means that the approximate expression (\ref{ratio}) is quite accurate. Furthermore, the data offloading ratio increases with the number of users, which is brought by the increasing aggregate cache capacity and the content sharing via D2D links. Fig. \ref{fig_scheme1-2} validates the accuracy of the approximation in (\ref{ratio}) by varying the file request parameter, where a larger $\gamma_r$ means that the requests are more concentrated on the popular files. It is shown that the approximate expression (\ref{ratio}) is quite accurate with different values of $\gamma_r$, and a larger value of $\gamma_r$ brings a higher performance gain of caching.

\subsection{Effect of User Mobility} 
In this part, we show the effect of user mobility, assuming the random caching strategy \cite{randomcache}. As in Section IV, we change the user moving speed by $\mu$ times. It has been shown in \cite{aggregate_real} that Gamma distribution can well fit the reciprocal of the aggregate average inter-contact time. Thus, the inter-contact parameters $\lambda^I_{i,j}, i \in \mathcal{S}, j \in \mathcal{S} \backslash \{i\}$ are generated as $\mu \widehat{\lambda^{I}_{i,j}}$, where $\widehat{\lambda^{I}_{i,j}}$ follows a gamma distribution as $\Gamma(4.43,1/1088)$ \cite{aggregate_real}. Similarly, the contact parameters are generated as $\mu \widehat{\lambda^{C}_{i,j}}$, where $\widehat{\lambda^{C}_{i,j}}$ follows $\Gamma(4.43 \times 25,1/1088/5)$.


\begin{figure}[t]
	\centering
	\includegraphics[width=3.5 in]{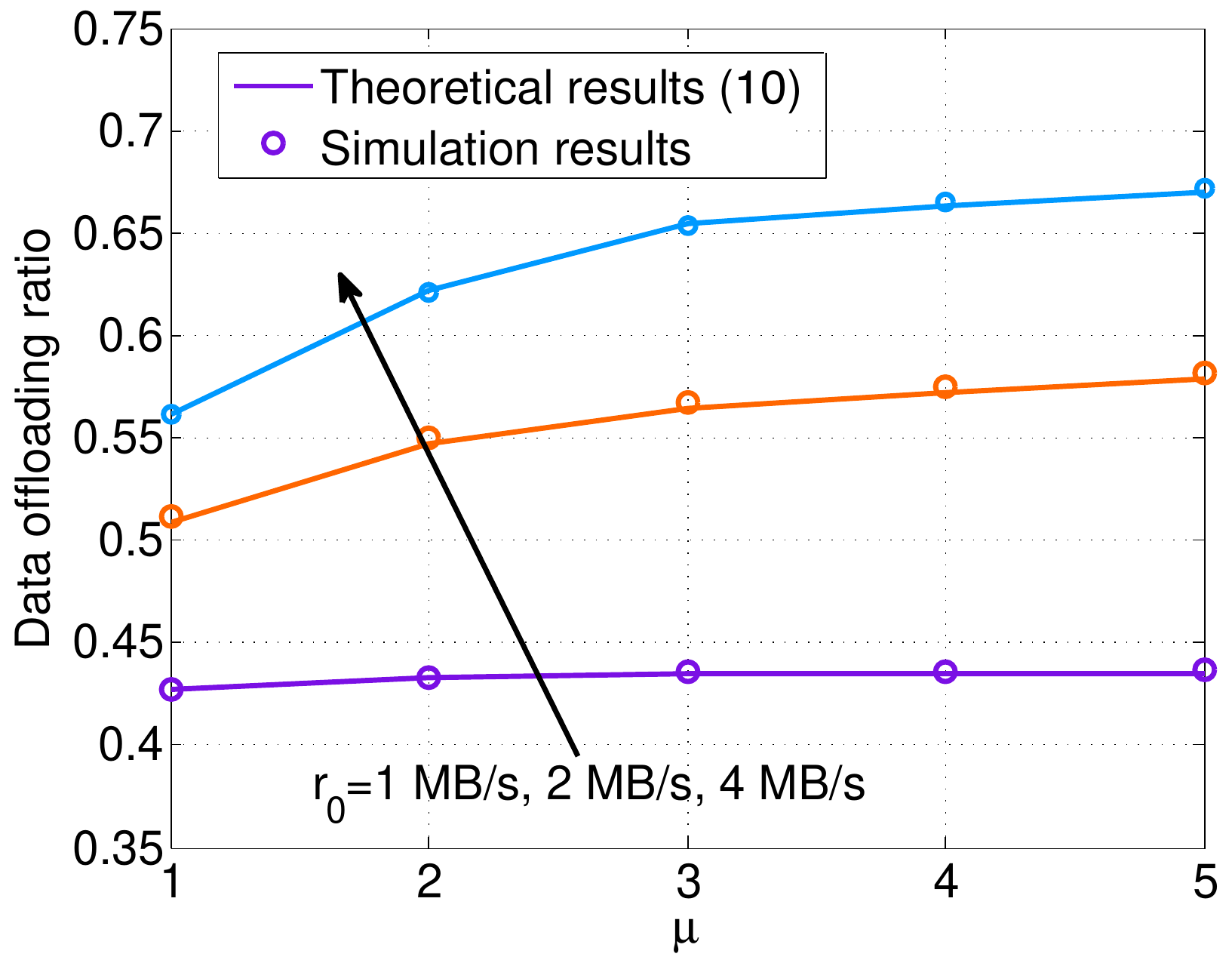}
	\caption{Data offloading ratio with $N_u=15$, $N_f=100$, $\tau_0 =120$ s, $\gamma_r=0.6$, $C=1$ GB, $F=100$ MB.}
	\label{fig_scheme2-2}
\end{figure}

In Fig. \ref{fig_scheme2-2}, the effect of $\mu$ is demonstrated. Firstly, the small gap between the theoretical and simulation results again verifies the accuracy of the approximate expression in (\ref{ratio}). It is also shown that the data offloading ratio increases with $\mu$, which confirms the conclusion in Proposition \ref{s_d}. We also observe that the increasing rate of the data offloading ratio decreases with the user moving speed, and increases with the data transmission rate.

\subsection{Mobility-Aware Caching Strategy}
In the following, we evaluate the mobility-aware caching strategy proposed in Section V. The following five caching strategies are compared.
\begin{itemize}
	\item Optimal mobility-aware caching: The optimal solution of problem (\ref{problem}), which is obtained by a DP algorithm similar to Algorithm 2 in  \cite{mobilitycaching}.\footnote{After making two changes, the DP algorithm in \cite{mobilitycaching} can be applied directly. The first one is to make the number of encoded segments of each file (i.e., $K_f$, $f \in \mathcal{N}_f$, in  \cite{mobilitycaching}) equals $1$. The other one is to change the utility function in  \cite{mobilitycaching} as $ U_f(x_{1,f}, \dots , x_{N_u,f} )=  \frac{1}{N_u}\sum \limits_{i \in \mathcal{S}} p^r_{i,f} \left[ x_{i,f} +  (1-x_{i,f}) \mathcal{R}_{i,f} \right]$.}
	\item Greedy mobility-aware caching: The suboptimal solution of problem (\ref{problem}) proposed in Section V.
	\item Greedy mobility-aware caching ignoring contact duration: This caching strategy only utilizes the information of inter-contact durations, while assuming that the whole file can be transmitted within one contact. This strategy is obtained by the suboptimal algorithm proposed in \cite{mobilitycaching}.
	\item Random caching: Each file is randomly cached by each user, where the caching probability is proportional to the file request probability \cite{randomcache}. 
	\item Popular caching: All the users cache the most popular files \cite{pupolarcaching}.
\end{itemize}

\begin{figure}
		\centering
		\includegraphics[width=3.5in]{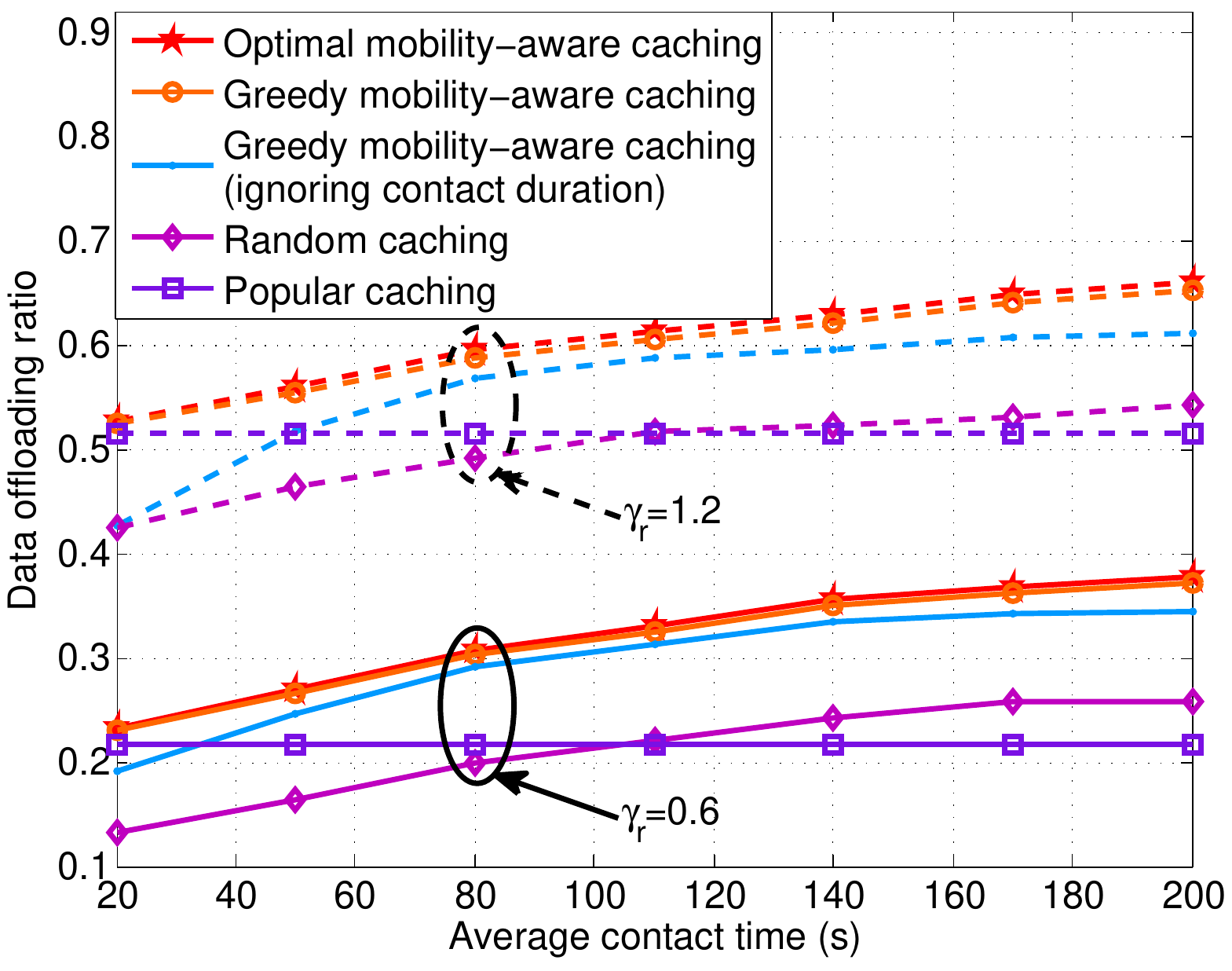}
		\caption{Data offloading ratio with $N_u=5$, $N_f=50$, $\tau_0 =300$ s, $C=1$ GB, $F=300$ MB, $r_0= 1.5$ MB/s.}
		\label{fig_scheme3-2}
\end{figure}

 
 In Fig. \ref{fig_scheme3-2}, we vary the average contact duration, denoted as $\bar{t^C}$, where the contact parameters $\lambda_{i,j}^C$ follow a gamma distribution, with expectation $1/\bar{t^C}$ and the same variance as $\Gamma(4.43,1/1088)$. 
 We observe that the performance of the greedy algorithm is quite close to the optimal one. Meanwhile, the mobility-aware caching strategies outperform the random and popular caching strategies, which demonstrates the advantage of exploiting the mobility information. Moreover, with the information of both the contact and inter-contact durations, the data offloading ratio can be further improved compared to the case ignoring the variation in contact durations. The conclusions do not change with different file request parameters. 
 Furthermore, the gap between the mobility-aware caching strategies considering and ignoring the contact durations first decreases and then increases with the average contact duration. When the contact duration is short, only part of a file can be delivered within one contact, and thus, it is critical to take the contact duration into account. On the other hand, when ignoring the contact duration, it is implicitly assumed that, a user is in the inter-contact duration with all the others when it requests a file. However, when the contact duration becomes larger, it is more likely that a user can start to download the file once the request is generated. Therefore, it is of critical importance to consider the contact duration, when it is comparable to the inter-contact duration. We also observe that, when the contact duration is very short, the popular caching strategy outperforms the mobility-aware caching strategy that ignores the contact duration, and approaches the proposed mobility-aware caching strategy. In such scenarios, the data shared via D2D links is quite limited, and the cached content is mainly used for a user's own need, and thus, popular caching is a good choice.
 
 We then evaluate the performance of the proposed mobility-aware caching strategy on a real-life data set collected during the INFOCOM 2016 conference \cite{CHSGCD07Impact}. There are $78$ selected conference participants, and each is distributed with one iMote, which is a Bluetooth radio device with a transmission range approximately as 30 meters. It is observed that the mobility pattern is quite different during the daytime and nighttime, i.e., users are much more frequently in contact with each other during the daytime. In the simulation, using the daytime data in the first day, we estimate the contact and inter-contact parameters (i.e., $\lambda^C_{i,j}$ and $\lambda^I_{i,j}$) by the inverse of the average contact and inter-contact durations, respectively. Then, different caching strategies are designed using the estimated contact and inter-contact parameters, and the performance on the daytime in the second day is shown in Fig. \ref{fig_scheme4-1}. We see that the proposed mobility-aware caching strategy outperforms the popular caching and random caching strategies by $5  \% \sim 35\%$ and $16 \% \sim 116\%$, respectively. The performance of the mobility-aware caching strategy that ignores the contact duration is also provided to show that the caching performance can be further improved with the information of contact durations. Moreover, the theoretical analysis of the data offloading ratio of the greedy mobility-aware caching strategy, which is calculated by (\ref{ratio}), is also shown in Fig. \ref{fig_scheme4-1}. The small gap between the simulation results and the theoretical value calculated by (\ref{ratio}) demonstrates that the alternating renewal process model and the approximated expression of the data offloading ratio in (\ref{ratio}) can provide a good approximation of practical performance.
 
 In this work, we simply the communication model for content delivery. To consider the performance in more realistic scenarios, we next investigate the effect of limited radio resources, which will limit the number of simultaneous communicating D2D pairs. We assume that each user can only serve one user's request at each time instant and there are in total 15 resource blocks allocated for D2D communications. For the case that one user receives multiple requests, it will randomly choose one to serve. When the number of D2D links is larger than the number of resource blocks, we will randomly choose 15 D2D links to transmit the requested files. Each user is assumed to make a new request right after the delay threshold of the previous request. The result is shown in Fig. \ref{fig_scheme4-1}, from which we see that the limited radio resources have very little effect on the performance. This is because the number of simultaneous D2D communications links is small, due to the randomness in user mobility.
 
 We also validate the performance of the proposed mobility-aware caching strategy in a campus scenario, based on a data-set collected in the Cambridge campus \cite{cam-leguay06}. The contact and inter-contact parameters (i.e., $\lambda^C_{i,j}$ and $\lambda^I_{i,j}$) are estimated based on the daytime data during the first three days. The performance on the fourth day is shown in Fig. \ref{fig_scheme4-2}. It again verifies the accuracy of the alternating renewal process model and the approximated expression of the data offloading ratio in (\ref{ratio}), and shows the inconspicuous effect of limited radio resources. We also observe that the proposed mobility-aware caching strategy outperforms the popular caching and random caching strategies by $6  \% \sim 25\%$ and $20 \% \sim 100\%$, respectively.

\begin{figure}
	\begin{minipage}[t]{0.49\linewidth}
		\centering
		\includegraphics[width=3in]{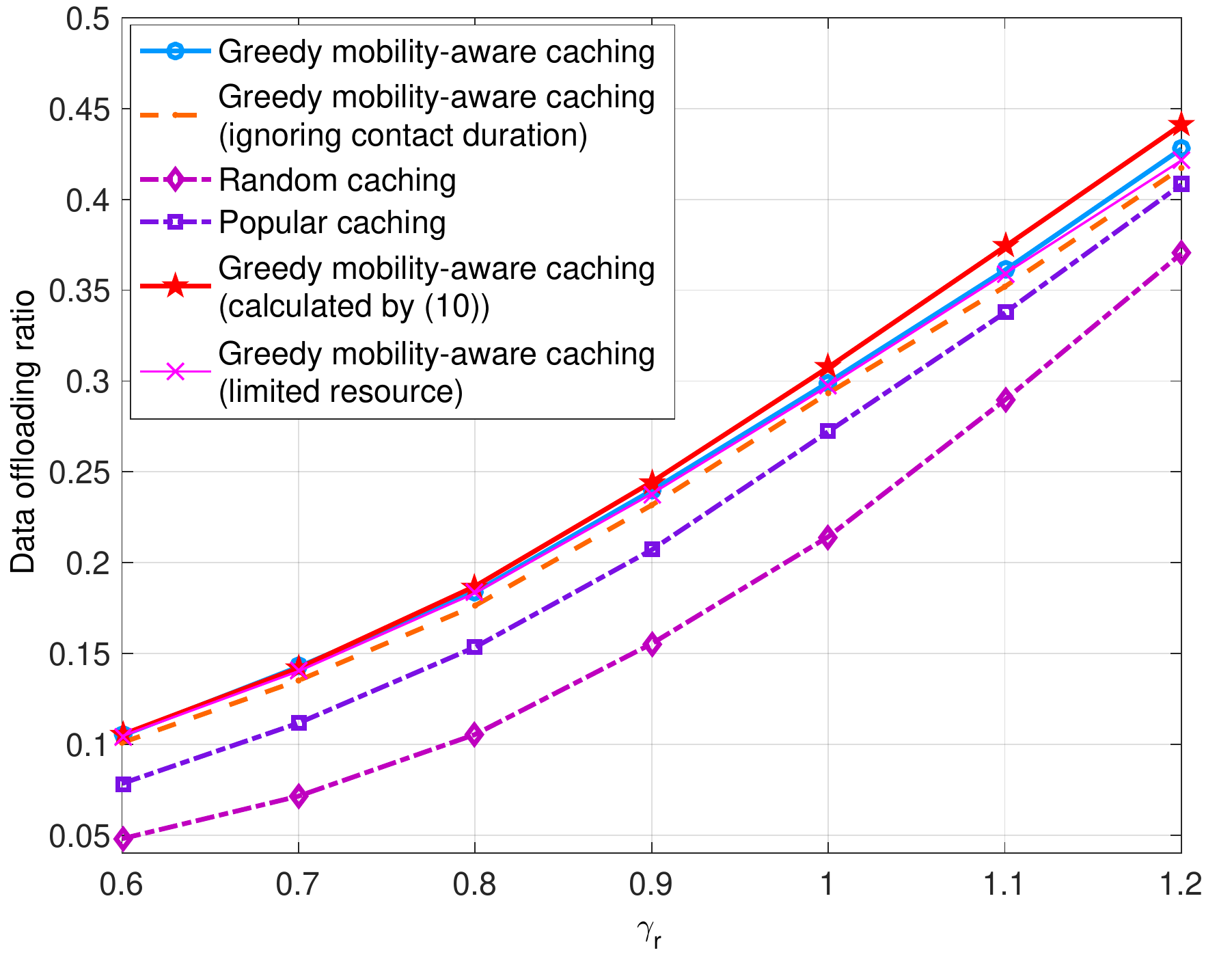}
		\caption{Data offloading ratio with $N_u=78$, $N_f=500$, $\tau_0 =300$ s, $C=1$ GB, $F=300$ MB, $r_0 = 2$ MB/s.}
		\label{fig_scheme4-1}
	\end{minipage}%
	\hspace{0.1in}
	\begin{minipage}[t]{0.49\linewidth}
		\centering
		\includegraphics[width=3in]{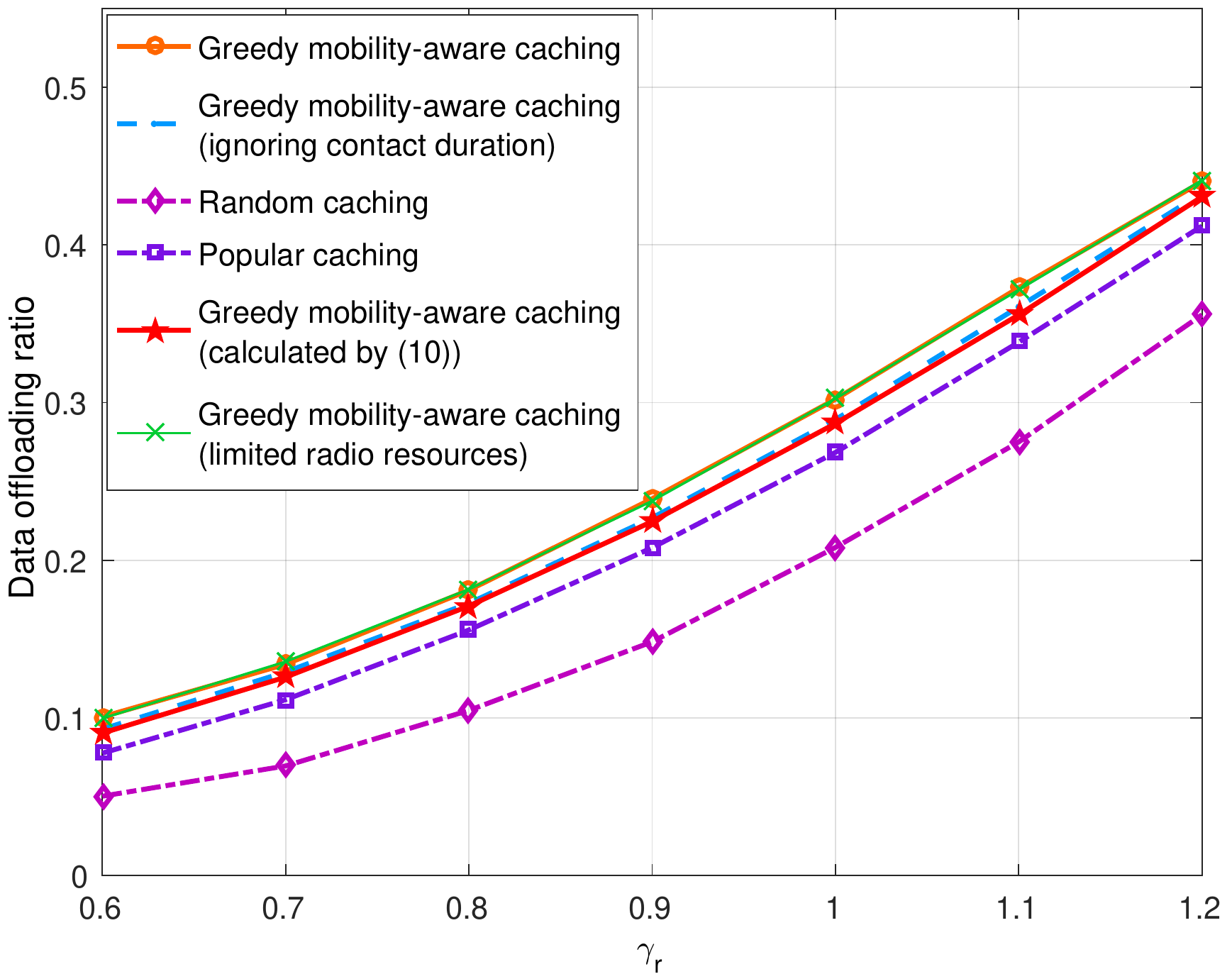}
		\caption{Data offloading ratio with $N_u=36$, $N_f=500$, $\tau_0 =600$ s, $C=1$ GB, $F=300$ MB, $r_0 = 1$ MB/s.}
		\label{fig_scheme4-2}
	\end{minipage}
\end{figure}


\section{Conclusions}

In this paper, we investigated a D2D caching network with mobile users. A tractable expression of the data offloading ratio was firstly derived and then used to prove that the data offloading ratio increases with the user speed. This result is valid in the scenario with low-to-medium speeds, where the transmission rate does not change with the user moving speed. The extension to the case with varying transmission rates is interesting for future investigation. The analytical results were also applied to develop a mobility-aware caching strategy by utilizing the statistical contact and inter-contact information. Simulation results validated the accuracy of the approximate expression of the data offloading ratio, and demonstrated that the data offloading ratio increases with the user speed, while the increasing rate decreases with the user speed. Moreover, we also observed that it is more critical to take the contact durations into account when they are relatively short or comparable to the inter-contact durations. One future direction is to develop online caching strategies in such systems. It is also interesting to consider more sophisticated resource allocation schemes and rate adaptation, as well as smarter user selection, during the content delivery phase.

\section*{Appendix}
\subsection{Proof of Lemma \ref{ev}}
As the timelines of different user pairs are independent, the expectation of the communication duration when user $i$ requests file $f$, which is not in its own cache, can be expressed as 
\begin{equation}
\mathbb{E} [ \tau^c_{i,f} ]=\lim \limits_{T^r_{i,f} \to \infty} \int _{T^r_{i,f}}^{T^r_{i,f}+\tau_0 } \left[ 1- \prod \limits_{j \in \mathcal{S},x_{j,f}=1}\left(1- \mathbb{E} H_{i,j}(t) \right) \right] dt. 
\end{equation}
Since the timeline between each pair of users is modeled as an alternating renewal process, according to Chapter 7 in \cite{renewalprocess}, we have $\lim \limits_{t \to \infty} \Pr[H_{i,j}(t)=1]=\frac{\lambda^I_{i,j}}{\lambda^C_{i,j}+\lambda^I_{i,j}}$, denoted as $p^I_{i,f}$. Thus, $\lim \limits_{t \to \infty} \mathbb{E}[ H_{i,j}(t)]=p^I_{i,f}$, and then, the expectation in (\ref{expectation_t}) can be obtained.
The variance of the communication duration is 
\begin{align} \label{v}
\mathrm{Var} [ \tau^c_{i,f} ]=
2 \lim \limits_{T^r_{i,f} \to \infty}  \int_{T^r_{i,f}}^{T^r_{i,f}+\tau_0 } \int_{T^r_{i,f}}^{\theta} \Pr[H^f_{i}(\theta)=1,H^f_{i}(t)=1] dt d\theta -\left( \mathbb{E} [ \tau^c_{i,f} ] \right)^2.
\end{align}
According to Chapter 7 in \cite{renewalprocess}, $\Pr[H_{i,j}(\theta)=0|H_{i,j}(t)=0]=p^I_{i,f}+p^I_{i,f}e^{-({\lambda^C_{i,j}+\lambda^I_{i,j}})(\theta-t)}$. Then, when $T_{i,f}^r \to \infty$, we get
\begin{align} \label{prob}
& \Pr[H^f_{i}(\theta)=1,H^f_{i}(t)=1] \notag \\
=&\Pr[H^f_{i}(\theta)=1]-\Pr[H^f_{i}(\theta)=1,H^f_{i}(t)=0] \notag \\
=&1-2 \prod \limits_{j \in \mathcal{S},x_{j,f}=1}p^I_{i,f} + \prod \limits_{j \in \mathcal{S},x_{j,f}=1} p^I_{i,f}\left[ p^I_{i,f} + \left( 1-p^I_{i,f}\right) e^{-(\lambda_{i,j}^C+\lambda^I_{i,j})(\theta-t)} \right]
\end{align}
Let $u=\theta-t$ and substitute (\ref{prob}) into (\ref{v}), then we get (\ref{var_t}).

\subsection{Proof of Proposition \ref{E_od}}
Based on the beta approximation of the communication duration, we approximate the data offloading ratio of user $i \in \mathcal{S}$ requesting file $f \in \mathcal{F}$ by 
\begin{equation} 
\mathcal{R}^a_{i,f}=\mathbb{E} _{Z_{i,f}}\left[ \min \left( r_0 \tau_0  Z_{i,f}/F ,1 \right) \right].
\end{equation}
Since $Z_{i,f} \sim \text{Beta} \left( \alpha_{i,f}, \beta_{i,f} \right)$, we have
\begin{align} 
&\mathcal{R}^a_{i,f} \notag\\
&= \frac{r_0 \tau_0 }{F}  \int_{0}^{\frac{F}{r_0 \tau_0 }} \frac{\left( z \right)^{\alpha_{i,f}} \left(1- z \right)^{(\beta_{i,f}-1)}}{B\left( \alpha_{i,f} ,\beta_{i,f}\right)}  d z + \int_{\frac{F}{r_0 \tau_0 }}^{1} \frac{\left( z \right)^{\left(\alpha_{i,f}-1\right)} \left(1- z \right)^{\left(\beta_{i,f}-1\right)} }{B\left( \alpha_{i,f}, \beta_{i,f}\right)} d z, \notag \\
& = \frac{r_0 \tau_0 }{F} \cdot \frac{\alpha_{i,f}}{\alpha_{i,f}+\beta_{i,f}}\int_{0}^{\frac{F}{r_0 \tau_0 }}  \frac{\left( z \right)^{\alpha_{i,f}} \left(1- z \right)^{(\beta_{i,f}-1)}}{B\left( \alpha_{i,f} +1 ,\beta_{i,f}\right)}  d z \notag \\
& \quad \quad + \int_{\frac{F}{r_0 \tau_0 }}^{1} \frac{\left( z \right)^{\left(\alpha_{i,f}-1\right)} \left(1- z \right)^{\left(\beta_{i,f}-1\right)} }{B\left( \alpha_{i,f}, \beta_{i,f}\right)} d z, \notag \\
& = \frac{r_0 \tau_0 }{F} \cdot \frac{\mathbb{E}[\tau^c_{i,f}]}{\tau_0 }I_{\frac{F}{r_0 \tau_0 }}\left( \alpha_{i,f}+1,\beta_{i,f}\right)+1-I_{\frac{F}{r_0 \tau_0 }}\left( \alpha_{i,f},\beta_{i,f}\right) , \notag \\
& = 1-I_{\frac{F}{\tau_0 r_0}} \left (\alpha_{i,f},\frac{\tau_0 -\mathbb{E}[\tau^c_{i,f}]}{\mathbb{E}[\tau^c_{i,f}]} \alpha_{i,f} \right)
+\frac{\mathbb{E}[\tau^c_{i,f}]r_0}{F} I_{\frac{F}{\tau_0 r_0}} \left(\alpha_{i,f}+1,\frac{\tau_0 -\mathbb{E}[\tau^c_{i,f}]}{\mathbb{E}[\tau^c_{i,f}]} \alpha_{i,f} \right).
\end{align}
Thus, we get the expression of the data offloading ratio as shown in Proposition \ref{E_od}.

\subsection{Proof of Lemma \ref{evv}}
When the contact and inter-contact parameters are scaled by $\mu$, $p^I_{i,j}(\mu)=\frac{\mu\lambda^C_{i,j}}{\mu\left( \lambda_{i,j}^C+\lambda_{i,j}^I \right)}=\frac{\lambda^C_{i,j}}{\lambda_{i,j}^C+\lambda_{i,j}^I }$ remains the same. Thus the expectation in (\ref{expectation_t}) does not change. Denote $a_0$, $a_{\mathbf{Z}}$, and $\kappa_{\mathbf{Z}}$ as in (\ref{ax}), and we get
\begin{align} 
\mathrm{Var} \left[ \tau^c_{i,f} \right] &= 2 \int_0^{\tau_0 } (\tau_0 -u) \left( a_0 +\sum \limits_{\mathbf{Z} \in \{0,x_{l,f}\}^{N_f} \backslash \mathbf{0}}  a_{\mathbf{Z}} e^{-\kappa_{\mathbf{Z}u}}\right) du -(\tau_0 )^2 a_0, \notag \\
&=2 a_0 \sum \limits_{\mathbf{Z} \in \{0,x_{l,f}\}^{n_f} \backslash \mathbf{0}} \frac{a_{\mathbf{Z}}}{\kappa_{\mathbf{Z}}} \left( \tau_0 -\frac{1}{\kappa_{\mathbf{Z}}} +\frac{1}{\kappa_{\mathbf{Z}}} 
\exp(-\kappa_{\mathbf{Z}}\tau_0 ) \right).
\end{align}
Then, by scaling the contact and inter-contact parameters, we get the result in (\ref{v_ts}).

\subsection{Proof of Lemma \ref{s_t}}

As shown in Lemma \ref{evv}, when the user speed changes by $\mu$ times, the expectation of the communication duration in (\ref{expectation_t}) does not change, while the variance changes. To prove that $\mathrm{Var}[\tau^c_{i,f}(\mu)]$ decreases with $\mu$, we will prove that $\frac{
\partial \mathrm{Var}[\tau^c_{i,f}(\mu)]}{\partial \mu}<0$. The partial derivation of $\mathrm{Var}[\tau^c_{i,f}(\mu)]$ is
\begin{align} \label{dir}
&\frac{\partial \mathrm{Var}[\tau^c_{i,f}(\mu)]}{\partial \mu}=  2 a_0 \sum \limits_{\mathbf{Z} \in \{0,x_{l,f}\}^{n_f} \backslash \mathbf{0}} \frac{a_{\mathbf{Z}}}{\mu^3\kappa_{\mathbf{Z}}} \mathcal{A}_1(x_{\mathbf{z}}),
\end{align}
where
$\mathcal{A}_1(x_{\mathbf{z}})=-x_{\mathbf{z}}-x_{\mathbf{z}} e^{-x_{\mathbf{z}}}-2(e^{-x_{\mathbf{z}}}-1)$ and $x_{\mathbf{z}}  \triangleq \mu \kappa_{\mathbf{Z}} \tau_0 > 0$.
Since $\mathcal{A}'_1(x_{\mathbf{z}})=-1+(1+x_{\mathbf{z}}) e^{-x_{\mathbf{z}}} < -1+(1+x_{\mathbf{z}}) \frac{1}{1+x_{\mathbf{z}}}=0$, $\mathcal{A}_1(x_{\mathbf{z}})$ is a decreasing function of $x_{\mathbf{z}}$. Thus, $\mathcal{A}_1(x_{\mathbf{z}}) < \mathcal{A}_1(0)=0$. According to (\ref{dir}), when $\exists 0<\lambda_{i,l}^C,\lambda_{i,l}^I<\infty \text{ and } x_{l,f}=1$, where $l \in \mathcal{S}$, there exists $ 0<\kappa_{\mathbf{Z}}<\infty$, where $\mathbf{Z} \in \{0,x_{l,f}\}^{n_f} \backslash \mathbf{0}$.
Thus, we have $\frac{
\partial \mathrm{Var}[\tau^c_{i,f}(\mu)]}{\partial \mu}<0$. The parameter $\alpha_{i,f}$ given in (\ref{beta_p}) is a decreasing function of $\mathrm{Var}[\tau^c_{i,f}(\mu)]$, and thus, it increases with $\mu$.

\subsection{Proof of Lemma \ref{s_g}}
To simplify the expression in (\ref{ratio_g}), denote $q \triangleq \frac{F}{\tau_0  r_0} \in (0,1)$, $y \triangleq \frac{\tau_0 -\mathbb{E}[\tau^c_{i,f}]}{\mathbb{E}[\tau^c_{i,f}]} \ge 0$, and $\alpha \triangleq \alpha_{i,f}$. The expression in (\ref{ratio_g}) can be rewritten as a function of $\alpha$ as
\begin{align}
	&\mathcal{R}_{i,f}=1- \frac{\int_0^q (1-\frac{u}{q}) u^{\alpha-1} (1-u)^{y \alpha-1} du }{B(\alpha,y\alpha)},
\end{align}
where $B(\cdot , \cdot)$ is the beta function. Let $g(\alpha)=1-\mathcal{R}_{i,f}$ with the derivative of $g(\alpha)$ given by
\begin{align}
	g'(\alpha)=
	&\frac{1}{B(\alpha,y\alpha)} \Bigg\{\int_0^q (1-\frac{u}{q}) u^{\alpha-1} (1-u)^{y \alpha-1}[\ln u + y \ln (1-u)] du \notag \\
	&- \int_0^q (1-\frac{u}{q}) u^{\alpha-1} (1-u)^{y \alpha-1} du D(y,\alpha) \Bigg\},
\end{align}
where $D(y,\alpha)=\psi(\alpha)+y \psi(y\alpha)-(1+y)\psi[(1+y)\alpha]$ and $\psi(\cdot)$ is the digamma function \cite{abramowitz1964handbook}. If $q=1$, $g'(\alpha)=\frac{\partial [y/(1+y)]}{\partial \alpha}=0$.
Denote $\mathcal{A}_2(q)=\frac{B(\alpha,y\alpha)}{q}g'(\alpha)$, and then we have $\mathcal{A}_2(1)=0$ and
\begin{align} 
	\lim \limits_{q \to 0^{+}} \mathcal{A}_2(q)= \lim \limits_{q \to 0^{+}} \int_0^q (q-u) u^{\alpha-1} (1-u)^{y \alpha-1}[\ln u + y \ln (1-u)] du
\end{align}
Since $ q \ge u \ge 0$ and $y \ge 0$, $(q-u) u^{\alpha-1} (1-u)^{y \alpha-1} \ge 0$ and $\ln u + y \ln (1-u) \le 0$, thus, $\lim \limits_{q \to 0^{+}} \mathcal{A}_2(q) \le 0$. The derivative of $\mathcal{A}_2(q)$ is 
\begin{align}
	\mathcal{A}'_2(q)= &\int_0^q u^{\alpha-1} (1-u)^{y \alpha-1}[\ln u + y \ln (1-u)] du \notag \\
	&- \int_0^q u^{\alpha-1} (1-u)^{y \alpha-1} du D(y,\alpha).
\end{align}
Thus, $\mathcal{A}'_2(1)=\frac{\partial B(\alpha,y \alpha)}{\partial \alpha}-\frac{\partial B(\alpha,y \alpha)}{\partial \alpha}=0$ and $\lim \limits_{q \to 0^{+}} \mathcal{A}'_2(q) \le 0$.
Then, we get $\mathcal{A}''_2(q)= q^{\alpha-1} (1-q)^{y \alpha-1}[\ln q + y \ln (1-q)-D(y,\alpha)]$. Let $\mathcal{A}_3(q)=q^{1-\alpha} (1-q)^{1-y \alpha} \mathcal{A}''_2(q)$, then, there is one zero point of $\mathcal{A}'_3(q)=\frac{1-(1+y)q}{q(1-q)}$ in $(0,1]$. Thus, there is one inflection point of $\mathcal{A}_3(q)$. Considering that $\lim \limits_{q \to 0^{+}}\mathcal{A}_3(q)=\lim \limits_{q \to 1^{-}}\mathcal{A}_3(q)=-\infty$, the sign of $\mathcal{A}_3(q)$ may be negative, or first negative, then positive, and then negative, when $q$ increases in $(0,1)$. If $\mathcal{A}_3(q)<0$, then $\mathcal{A}''_2(q)<0$, when $q \in (0,1)$. However, we have $\lim \limits_{q \to 0^{+}} \mathcal{A}'_2(q) \le \mathcal{A}'_2(1)$, which means that $\mathcal{A}'_2(q)$ cannot be a decreasing function in $(0,1)$. Thus, the sign of $\mathcal{A}_3(q)$ is first negative, then positive, and then negative, when $q$ increases in $(0,1)$. Since $\mathcal{A}''_2(q)$ has the same sign with $\mathcal{A}_3(q)$ in $(0,1)$, $\mathcal{A}'_2(q)$ first decreases, then increases, and then decreases when $q$ increases in $(0,1)$. Considering that $\lim \limits_{q \to 0^{+}} \mathcal{A}'_2(q) \le 0$ and $\mathcal{A}'_2(1)=0$, the sign of $\mathcal{A}'_2(q)$ must be first negative, and then positive in $(0,1)$. Therefore, when $q$ increases in $(0,1)$, $\mathcal{A}_2(q)$ first decreases, and then increases. Considering that $\lim \limits_{q \to 0^{+}} \mathcal{A}_2(q) \le 0$ and $\mathcal{A}_2(1)=0$, we have $\mathcal{A}_2(q)<0$ in $(0,1)$. Since $g'(\alpha)=\frac{q}{B(\alpha,y\alpha)} \mathcal{A}_2(q)$, we get $g'(\alpha)<0$ in $(0,1)$. Thus, $g(\alpha)$ decreases with $\alpha$, and $\mathcal{R}_{i,f}=1-g(\alpha)$ increases with $\alpha$.

\subsection{Proof of Proposition \ref{s_d}}
The data offloading ratio in (\ref{ratio}) increases with $\mathcal{R}_{i,f}$ if $x_{i,f}=0$, $i \in \mathcal{S}$, $f \in \mathcal{F}$. Thus, based on Lemmas  \ref{s_t} and \ref{s_g}, we get that the data offloading ratio when user $i$ requests file $f$ from other users (i.e., $\mathcal{R}_{i,f}$) increases with the user speed when $\exists j \in \mathcal{S}$ such that $0<\lambda_{i,j}^C,\lambda_{i,j}^I<\infty \text{ and } x_{j,f}=1$. Otherwise, $\mathcal{R}_{i,f}=0$. Accordingly, the data offloading ratio when user $i$ requests file $f$ (i.e., $x_{i,f}+ (1-x_{i,f}) \mathcal{R}_{i,f} $) increases with the user speed when $x_{i,f}=0$, and $\exists j \in \mathcal{S}$ such that $0<\lambda_{i,j}^C,\lambda_{i,j}^I<\infty \text{ and } x_{j,f}=1$. Otherwise, it remains the same. Since we consider that there exists a pair of users $i,j \in S, i\neq j$ with $0 < \lambda_{i,j}^I,\lambda_{i,j}^C  < \infty$, and a file $f \in \mathcal{F}$, such that $x_{i,f}=1$ and $x_{j,f}=0$, the data offloading ratio increases with the user speed.

\subsection{Proof of Lemma \ref{sub_p}}
To prove $\mathcal{R}(Y)$ is a monotone submodular set function, we will prove that it satisfies Proposition \ref{secondorder}. In the following, we will first prove that $\mathcal{R}_{i,f}(Y)$ in (\ref{rifY}) satisfies Proposition \ref{secondorder}. Let $A \subseteq S$, $y_{j_1,f},y_{j_2,f} \in S-A$, and $j_1 \ne j_2$, and then we have
\begin{align}
& \mathcal{R}_{i,f}(A \cup\{y_{j_2,f}\})-\mathcal{R}_{i,f}(A) \notag \\
=& \mathbb{E} \left[ \min \left( r_0 \tau^c_{i,f}(A \cup\{y_{j_2,f}\})/F ,1 \right) \right]-\mathbb{E} \left[ \min \left( r_0 \tau^c_{i,f}(A)/F ,1 \right) \right] \notag \\
= & \mathbb{E}\left[ \min \left( r_0 \tau^c_{i,f}(A)/F + r_0 D_2/F ,1 \right) - \min \left( r_0 \tau^c_{i,f}(A)/F ,1 \right) \right],
\end{align}
where $D_2$ is the duration that user $i$ is in contact with user $j_2$, and does not in contact with the users in $A$, given by
\begin{equation}
D_2=\lim\limits_{T^r \to \infty} \int _{T^r}^{T^r+ \tau_0 } \mathbbm{1} \left[ H_{i,j}(t)=0, \forall y_{j,f} \in A, \text{ and } H_{i,j_2}(t)=1 \right] dt.
\end{equation}
Then, we get 
\begin{align}
 \mathcal{R}_{i,f}(A \cup\{y_{j_2,f}\})-\mathcal{R}_{i,f}(A)
=
\begin{cases}
\mathbb{E} \left[r_0  D_2 /F \right] & \text{if } \tau^c_{i,f}(A)<F/r_0-D_2, \\
\mathbb{E} \left[1- r_0 \tau^c_{i,f}(A)/F \right] & \text{if } F/r_0-D_2 \leqslant \tau^c_{i,f}(A) \leqslant F/r_0, \\
0 & \text{if } \tau^c_{i,f}(A) > F/r_0.
\end{cases}
\end{align}
Thus, $ \mathcal{R}_{i,f}(A \cup\{y_{j_2,f}\})-\mathcal{R}_{i,f}(A) \geqslant 0$ and $ \mathcal{R}_{i,f}(A \cup\{y_{j_1,f}\})-\mathcal{R}_{i,f}(A) \geqslant 0$, i.e., $ \mathcal{R}_{i,f}(Y)$ is monotone.
Similarly, we have
\begin{align}
& \mathcal{R}_{i,f}(A \cup\{y_{j_1,f},y_{j_2,f}\})-\mathcal{R}_{i,f}(A\cup\{y_{j_1,f}\}) \notag\\ &=
\begin{cases}
\mathbb{E} \left[r_0  D'_2 /F \right] & \text{if } \tau^c_{i,f}(A)<F/r_0-D'_2-D_1, \\
\mathbb{E} \left[1- r_0 \tau^c_{i,f}(A)/F-r_0 D_1/F \right] & \text{if } F/r_0-D'_2-D_1 \leqslant \tau^c_{i,f}(A) \leqslant F/r_0-D_1, \\
0 & \text{if } \tau^c_{i,f}(A) > F/r_0-D_1,
\end{cases}
\end{align}
where $D_1$ is the duration that user $i$ is in contact with user $j_1$, and does not in contact with the users in $A$, given by
\begin{equation}
D_1=\lim\limits_{T^r \to \infty} \int _{T^r}^{T^r+ \tau_0 } \mathbbm{1} \left[ H_{i,j}(t)=0, \forall y_{j,f} \in A, \text{ and } H_{i,j_1}(t)=1 \right] dt,
\end{equation}
and $D'_2$ is the duration that user $i$ is in contact with user $j_2$, and does not in contact with the users in $A \cup \{y_{j_1,f}\}$, given by
\begin{equation}
D'_2=\lim\limits_{T^r \to \infty} \int _{T^r}^{T^r+ \tau_0 } \mathbbm{1} \left[ H_{i,j}(t)=0, \forall y_{j,f} \in A \cup \{y_{j_1,f}\}, \text{ and } H_{i,j_2}(t)=1 \right] dt.
\end{equation}
The following inequalities reveal the relationship among $D_1$, $D_2$, and $D'_2$. Frist, we have
\begin{equation}
D_2-D'_2=\lim\limits_{T^r \to \infty} \int _{T^r}^{T^r+ \tau_0 } \mathbbm{1} \left[ H_{i,j}(t)=0, \forall y_{j,f} \in A, \text{ and } H_{i,j_2}(t)=H_{i,j_1}(t)=1 \right] dt \geqslant 0,
\end{equation}
which is the duration that user $i$ is in contact with user $j_1$ and $j_2$, and does not in contact with the users in $A$. We also have 
\begin{equation}
D_1+D'_2-D_2=\lim\limits_{T^r \to \infty} \int _{T^r}^{T^r+ \tau_0 } \mathbbm{1} \left[ H_{i,j}(t)=0, \forall y_{j,f} \in A \text{ and } H_{i,j_1}(t)=1 \text{ and }H_{i,j_2}(t)=0 \right] dt \geqslant 0,
\end{equation} 
which is the duration that user $i$ is in contact with user $j_1$, and does not in contact with the users in $A\cup \{y_{j_2,f}\}$.
Let $\Delta=\left[\mathcal{R}_{i,f}(A \cup\{y_{j_2,f}\})-\mathcal{R}_{i,f}(A) \right] - \left[ \mathcal{R}_{i,f}(A \cup\{y_{j_1,f},y_{j_2,f}\})-\mathcal{R}_{i,f}(A\cup\{y_{j_1,f}\}) \right]$, and we get
\begin{align}
\Delta=
\begin{cases}
\mathbb{E} \left[r_0 \left(D_2- D'_2 \right) /F \right] & \text{if } \tau^c_{i,f}(A)<F/r_0-D'_2-D_1, \\
\mathbb{E} \left[r_0 \left(\tau^c_{i,f}(A)+D_1+D_2 \right)/F-1 \right] & \text{if } F/r_0-D'_2-D_1 \leqslant \tau^c_{i,f}(A) < F/r_0-D_2, \\
\mathbb{E} \left[r_0 D_1 /F \right] & \text{if } F/r_0-D_2 \leqslant \tau^c_{i,f}(A) < F/r_0-D_1, \\
\mathbb{E} \left[1- r_0 \tau^c_{i,f}(A)/F \right] &  \text{if } F/r_0-\min(D_1,D_2) \leqslant \tau^c_{i,f}(A) < F/r_0, \\
0 & \text{if } \tau^c_{i,f}(A) \geqslant F/r_0.
\end{cases}
\end{align}
Thus, when $\tau^c_{i,f}(A)<F/r_0-D'_2-D_1$, we have $\Delta=\mathbb{E} \left[r_0 \left(D_2- D'_2 \right) /F \right] \geqslant 0$. When $F/r_0-D'_2-D_1 \leqslant \tau^c_{i,f}(A) < F/r_0-D_2$, we have
\begin{align}
\Delta & =\mathbb{E} \left[r_0 \left(\tau^c_{i,f}(A)+D_1+D_2 \right)/F-1 \right] \notag \\
& \geqslant \mathbb{E} \left[r_0 \left(\tau^c_{i,f}(A)+D_1+D'_2 \right)/F-1 \right] \notag \\
& \geqslant \mathbb{E} \left[r_0 \left(F/r_0 - D'_2 -D_1 +D_1+D'_2 \right)/F-1 \right] \notag \\
& = 0.
\end{align}
When $F/r_0-\min(D_1,D_2) \leqslant \tau^c_{i,f}(A) < F/r_0$, we have $\Delta=\mathbb{E} \left[1- r_0 \tau^c_{i,f}(A)/F \right] \geqslant 0$. Thus, we get $\Delta \geqslant 0$, which means that
\begin{equation}
\mathcal{R}_{i,f}(A \cup\{y_{j_2,f}\})-\mathcal{R}_{i,f}(A) \geqslant \mathcal{R}_{i,f}(A \cup\{y_{j_1,f},y_{j_2,f}\})-\mathcal{R}_{i,f}(A\cup\{y_{j_1,f}\}).
\end{equation}
Then, we will prove that $\mathcal{R}(Y)$ satisfies Proposition \ref{secondorder}. 
Let $A \subseteq S$, $y_{j_1,f_1},y_{j_2,f_2} \in S-A$, and $y_{j_1,f_1} \ne y_{j_2,f_2}$, and we first consider the case that $f_1=f_2=f$, which gives
\begin{align}
&\mathcal{R}(A \cup\{y_{j_2,f}\})-\mathcal{R}(A) \notag \\
= & \frac{1}{N_u}\left\{ p^r_{j_2,f}\left[1-\mathcal{R}_{j_2,f}(A)\right] + \sum \limits_{i \in \mathcal{S}\backslash \{j_2\}} p^r_{i,f} 
\mathbbm{1}(y_{i,f} \notin A) \left[ \mathcal{R}_{i,f}(A \cup \{y_{j_2,f}\})- \mathcal{R}_{i,f}(A) \right] \right\}, \notag \\
= & \frac{1}{N_u}\Bigg\{ p^r_{j_2,f}\left[1-\mathcal{R}_{j_2,f}(A)\right] + \sum \limits_{i \in \mathcal{S}\backslash \{j_1,j_2\}} p^r_{i,f} 
\mathbbm{1}(y_{i,f} \notin A) \left[ \mathcal{R}_{i,f}(A \cup \{y_{j_2,f}\})- \mathcal{R}_{i,f}(A) \right] \notag \\
& \quad + p^r_{j_1,f} \left[ \mathcal{R}_{j_1,f}(A \cup \{y_{j_2,f}\})- \mathcal{R}_{j_1,f}(A) \right] \Bigg\}, \notag 
\end{align}
\begin{align}
\geqslant & \frac{1}{N_u}\Bigg\{ p^r_{j_2,f}\left[1-\mathcal{R}_{j_2,f}(A \cup \{y_{j_1,f}\})\right] \notag \\
& \quad + \sum \limits_{i \in \mathcal{S}\backslash \{j_1,j_2\}} p^r_{i,f} 
\mathbbm{1}(y_{i,f} \notin A) \left[ \mathcal{R}_{i,f}(A \cup\{y_{j_1,f}, y_{j_2,f}\})- \mathcal{R}_{i,f}(A\cup\{y_{j_1,f}\}) \right] \Bigg\}, \notag \\
= & \mathcal{R}(A \cup\{y_{j_1,f}, y_{j_2,f}\})-\mathcal{R}(A\cup\{ y_{j_1,f}\}) \geqslant 0.
\end{align}
For the case $f_1 \ne f_2$, we have
\begin{align}
&\mathcal{R}(A \cup\{y_{j_2,f_2}\})-\mathcal{R}(A) \notag \\
= & \frac{1}{N_u}\left\{ p^r_{j_2,f_2}\left[1-\mathcal{R}_{j_2,f_2}(A)\right] + \sum \limits_{i \in \mathcal{S}\backslash \{j_2\}} p^r_{i,f_2} 
\mathbbm{1}(y_{i,f_2} \notin A) \left[ \mathcal{R}_{i,f}(A \cup \{y_{j_2,f_2}\})- \mathcal{R}_{i,f_2}(A) \right] \right\}, \notag \\
= & \mathcal{R}(A \cup\{y_{j_1,f_1}, y_{j_2,f_2}\})-\mathcal{R}(A\cup\{ y_{j_1,f_1}\}) \geqslant 0.
\end{align}
Accordingly, $\mathcal{R}(Y)$ satisfies Proposition \ref{secondorder}, and thus, it is a monotone submodular set function.
\bibliographystyle{IEEEtran}
\bibliography{IEEEabrv,report}
\end{document}